\documentclass[12pt]{article}
\usepackage{graphicx}
\usepackage{amsfonts}
\usepackage{amsmath}
\usepackage{amssymb}
\usepackage{amsthm}
\usepackage{algorithm}
\usepackage{algpseudocode}
\usepackage{braket}
\usepackage{float}
\usepackage{subfigure}
\usepackage{hyperref}
\usepackage[all]{hypcap}
\usepackage{xfrac}
\usepackage{tikz}
\usepackage{centernot}
\usepackage{enumitem}
\usepackage{xcolor}
\usepackage{dsfont}
\usepackage{authblk}
\usepackage{bbm}

\newtheorem{lemma}{Lemma}
\newtheorem{theorem}{Theorem}
\newtheorem{definition}{Definition}

\usepackage[marginparwidth=2.8cm]{geometry} 
\usepackage{todonotes}

\newcommand{\norm}[1]{\left\lVert#1\right\rVert}
\newcommand{\floor}[1]{\lfloor #1 \rfloor}
\newcommand{\ceil}[1]{\lceil #1 \rceil}
\renewcommand{\mark}{\text{\upshape{mark}}}
\newcommand{\tile}{\text{\upshape{tile}}}
\newcommand{\uu}{\text{\upshape{u}}}
\newcommand{\col}{\text{\upshape{col}}}
\newcommand{\row}{\text{\upshape{row}}}
\newcommand{\site}{\text{\upshape{site}}}
\newcommand{\spec}{\text{\upshape{spec}}}
\newcommand{\dense}{\text{\upshape{dense}}}
\newcommand{\trivial}{\text{\upshape{trivial}}}
\newcommand{\tot}{\text{\upshape{tot}}}
\newcommand{\inn}{\text{\upshape{in}}}
\newcommand{\out}{\text{\upshape{out}}}
\newcommand{\prop}{\text{\upshape{prop}}}
\newcommand{\comp}{\text{\upshape{comp}}}
\newcommand{\anc}{\text{\upshape{anc}}}
\newcommand{\QTM}{\text{\upshape{QTM}}}

\title{Chaitin Phase Transition}

\author[1]{James Purcell \thanks{james.purcell.22@ucl.ac.uk}}
\author[2]{Zhi Li \thanks{zli@perimeterinstitute.ca}}
\author[1,3]{Toby Cubitt \thanks{t.cubitt@ucl.ac.uk}}
\affil[1]{University College London}
\affil[2]{Perimeter Institute for Theoretical Physics}
\affil[3]{Phasecraft Ltd.}
\date{} 

\begin{document}
	\maketitle
	
	\begin{abstract}
		We construct a family of Hamiltonians whose phase diagram is guaranteed to have a single phase transition, yet the location of this phase transition is uncomputable. 
		The Hamiltonians $H(\phi)$ describe qudits on a two-dimensional square lattice with translationally invariant, nearest-neighbour interactions tuned by a continuous parameter $\phi\in(0,1]$.
		For all $\phi\in(0,1]$, $H(\phi)$ is in one of two phases, one a gapless phase, the other a gapped phase. 
		The phase transition occurs when $\phi$ equals the Chaitin's constant $\Omega$, a well-defined real number that encodes the Halting problem, and hence is uncomputable for Turing machines and undecidable for any consistent recursive axiomatization of mathematics.
		Our result implies that no general algorithm exists to determine the phase diagrams even under the promise that the phase diagram is exceedingly simple, 
		and illustrates how uncomputable numbers may manifest in physical systems.
	\end{abstract}
	
	\section{Introduction}
	
	Understanding the macroscopic behaviour of many-body systems from their microscopic origin has long been a key focus of condensed matter physics.
	While the microscopic description usually takes the form of Hamiltonians, the collective, macroscopic properties that emerge from the microscopic interactions are characterized by phases and phase transitions in the thermodynamic limit.

	The study of phases and phase transitions from Hamiltonians has been quite successful.
	Certain systems may permit a complete analytical understanding, such as the quantum Ising model \cite{lieb1961two}, spin-1/2 Heisenberg model \cite{Heisenberg1928}, AKLT model \cite{AKLT}, Kitaev's toric code \cite{kitaev2003fault} and honeycomb model \cite{kitaev2006anyons}.
	Numerically, general methods for computing macroscopic properties of finite systems exist, such as exact diagonalization, Monte Carlo \cite{metropolis1953equation}, density matrix renormalization group algorithms \cite{dmrg}, and tensor networks \cite{fannes1992finitely}. 
	Thermodynamic limit properties, such as gapped v.s. gapless, are extrapolated by applying the algorithm of choice on large enough systems.
	For certain families of frustration-free Hamiltonians, such extrapolation can be placed on rigorous grounds \cite{knabe1988energy,bravyi-gosset, gosset-mozgunov, lemm-mozgunov}.
	Furthermore, for systems depending on external parameters, a phase diagram can often be obtained by repeating the above procedure whilst varying the free parameter.

	On the other hand, it has been recognized that there must be some fundamental difficulties in calculating the phase diagram of a given Hamiltonian in general.
	Stemming from Kitaev's observation that computing the ground state energy of a many-body Hamiltonian is QMA-complete \cite{kitaev}, the field of Hamiltonian complexity has been established as a critical area of research \cite{oliveira2005complexity,kempe2006complexity,aharonov2009power,GI09}.
	When passing to the thermodynamic limit of infinitely large system sizes, computability theory is often more relevant. 
	Cubitt, Perez-Garcia, and Wolf \cite{CPW,CPWfull} proved that the spectral gap problem is in general undecidable, 
	by constructing a family of 2D lattice Hamiltonian for which determining whether an instance is gapped or not is at least as difficult as the Halting problem, a prototypical undecidable problem in computer science \cite{turing1936computable}.
	Such undecidability also persists for 1D Hamiltonians \cite{Bausch_2020}.
	While the above Hamiltonians depend on discrete parameters, a subsequent work \cite{BCW} constructed a family of Hamiltonian $H(\phi)$ that depends continuously on a real parameter $\phi\in\mathbb{R}$. The resulting phase diagram is uncomputable, fractal-like, and has infinitely many phase transitions.

	An immediate question is how to reconcile the empirical successes with the fundamental hardness. Specifically, the Hamiltonians discussed in \cite{BCW} exhibit a highly complex phase diagram, which appears very different from what is observed in classic condensed matter models.	
	Can uncomputability appear in more ``physical" settings?
	More precisely, in the work, we ask
	\begin{quote}
		\textit{Can a phase diagram that is guaranteed to have the simplest possible form (i.e., a single phase transition) be uncomputable?}
	\end{quote}

	We answer this question affirmatively by constructing a one-parameter family of Hamiltonians that undergo a single phase transition at a well-defined, yet uncomputable point. 
	Interestingly, our construction ensures that the phase transition occurs at the \textit{Chaitin's constant} $\Omega_M$ \cite{chaitin} with respect to a universal Turing machine (UTM). 
	
	The Chaitin's constant $\Omega_M$ is an unusual example of a \textit{definable} yet \emph{uncomputable} number. 
	It may be interpreted as the probability that a UTM halts on randomly chosen inputs, and can be rigorously defined via a convergent infinite series (see Definition \ref{chaitin const})).
	It compactly encodes the whole Halting problem and is uncomputable: no Turing machine can compute its value to an arbitrary number of bits.
	In fact, its binary expansion is random in a strict sense \cite{chaitin,Downey_Hirschfeldt_2010} and is fundamentally unpredictable.
	Relatedly, it is undecidable in formal mathematical systems: any consistent recursive axiomatization of mathematics, such as the ZFC axioms of set theory---the most widely accepted foundation of mathematics---can determine only finitely many bits of $\Omega_M$ \cite{chaitin1992information}. 
	In extreme cases, one may be limited to knowing very little about it.
	For example \cite{solovay1999version}, there exists a universal Turing machine for which every bit of its Chaitin's constant is independent of the ZFC axioms.

	Our Hamiltonian, therefore, has an arguably simple phase diagram, similar to those found in ``real-world" physics. However, there is a fundamental barrier to computing these diagrams with arbitrarily high precision. 
	Our construction vividly illustrates how uncomputable numbers may manifest in physical systems.

	\section{Main Result}

	Our main result is an explicit construction of a family of 2D lattice Hamiltonians $\{H^{(L)}(\phi)\}$ indexed by the system size $L$ and a periodic tuning parameter $\phi\in\mathbb{R}$---i.e. $H^{(L)}(\phi)=H^{(L)}(\phi+1)$---such that in the thermodynamic limit the system is:
	\begin{itemize}
		\item gapless, disordered, critical if $\phi\in(0,\Omega_M)$;
		\item gapped, ordered, trivial if $\phi\in[\Omega_M,1]$.
	\end{itemize}
	Here, $\Omega_M\in(0,1)$ is the Chaitin's constant for a universal prefix-free (see Sec.~\ref{sec-Chaitin}) Turing machine $M$, and is hence uncomputable.
	The phase diagram is illustrated in Fig.~\ref{fig1}.
	\begin{figure}[H]
		\centering
		\begin{tikzpicture}
			\draw[-stealth] (0,0) -- (6.1,0) node[right]{$\phi$};
			
			\filldraw[red!40, thin] (0.05,0.05) rectangle (4,0.9);
			\filldraw[green!40, thin] (4,0.05) rectangle (6,0.9);
			
			\draw[thick] (0,0) node[below]{$0$};
			\draw[thick] (4.2,0) node[below]{$\Omega_M$};
			\draw[] (4,-0.35) node[below]{\text{(Chaitin)}};
			\draw[thick] (6,0) node[below]{$1$};
			
			\filldraw[red!40, thin] (-3,0.5) rectangle (-2.5,0.7);
			\node at (-1.5, 0.6) {gapless};
			\filldraw[green!40, thin] (-3,0.2) rectangle (-2.5,0.4);
			\node at (-1.5, 0.3) {gapped};
			
		\end{tikzpicture}
		\caption{The diagram plotted against $\phi\in\mathbb{R}$. The Hamiltonian is periodic and we only plot one fundamental period. 
			In the red (green) region the system is gapless (gapped), disordered (ordered), and critical (trivial). 
			$\Omega_M$ is the Chaitin's constant of a universal Turing machine $M$ and is uncomputable.
		}\label{fig1}
	\end{figure}
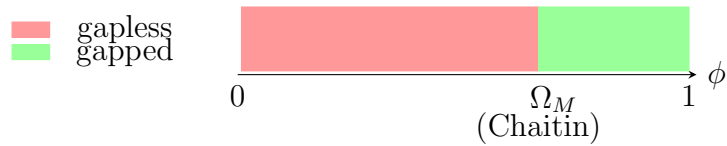
	
	Below, we present the precise result.
	We consider qudits living on a 2D square lattice $\Lambda$ of size $L\times L$. 
	We shall construct Hamiltonians $H^{(L)}(\phi)$ in the following form:
	\begin{equation} \label{hlambda}
		H^{(L)}(\phi):=\sum_{i=1}^{L} \sum_{j=1}^{L} h^{\site}_{(i,j)}
		+\sum_{i=1}^{L} \sum_{j=1}^{L-1} h^{\col}_{(i,j),(i,j+1)}
		+\sum_{i=1}^{L-1} \sum_{j=1}^{L} h^{\row}_{(i,j),(i+1,j)}(\phi).
	\end{equation}
	Here $h^{\site}$ acts on site, and $h^{\row}$ (and $h^{\col}$) only act on two horizontally (and vertically) adjacent sites. 
	All of them are independent of the location, hence $H^{(L)}$ is nearest-neighbour and translationally invariant. 
	The precise construction of the Hamiltonian is provided by the following:		
	\begin{theorem}\label{theorem 1}
		For any prefix-free	Turing machine $M$, we can explicitly construct a rational number $\beta\in(0,1]$, a dimension $d \in \mathbb{N}$, two $d\times d$ matrices $h$ and $h'$, five $d^2 \times d^2$ matrices $a,a',b,c,c'$, and define $H^{(L)}$ as in Eq. (\ref{hlambda}) by setting:
		\begin{align}
			h^{\site}&= h + \beta h',\\
			h^{\col}&= c + \beta c',\\
			h^{\row}(\phi) &= a + \beta \left( a' + e^{i \pi \phi} b + e^{-i \pi \phi} b^\dagger \right),
		\end{align}
		such that	
		\begin{itemize}
			\item matrix entries of $a,a',b,c,c',h, h'$ are in $\mathbb{Z} + \frac{1}{\sqrt{2}} \mathbb{Z}$;
			\item $||h^{\row}(\phi)||\leq 2$ for all $\phi$, $||h^{\col}||\leq 2$, $||h^{\site}||\leq 2$;
			\item for $\phi \text{\upshape{\text{ mod }}} 1 \in (0, \Omega_M)$, $H^{\Lambda}(\phi)$ is gapless with a disordered and critical ground state in the thermodynamic limit;
			\item for $\phi \text{\upshape{\text{ mod }}} 1 \in [\Omega_M, 1]$, $H^{\Lambda}(\phi)$ is gapped with an ordered and trivial ground state in the thermodynamic limit.
		\end{itemize}
	\end{theorem}

	In condensed matter physics, quantum phase transitions are often characterized by spectral gap closing, non-analytic changes in certain order parameters, and/or different infrared properties of the ground states. 
	As stated in the theorem, our construction satisfies all three characterizations.  
	
	Denote $\lambda_0(H^{(L)})$ and $\lambda_1(H^{(L)})$ as the minimal and the second minimal eigenvalue respectively.
	We say a family of Hamiltonians $\{H^{(L)}\}$ is \emph{gapped} if there is a constant lower bound on the spectral gap: $\Delta(H^{(L)})=\lambda_1(H^{(L)})-\lambda_0(H^{(L)})\geq \gamma>0$ (hence $\lambda_0(H^{(L)})$ is non-degenerate; in our model, $\gamma=1$) for all sufficiently large $L$.
	A family of Hamiltonians is \emph{gapless} if the system has a ``continuous spectrum" above the ground state in the thermodynamic limit:
	if there is a constant $c > 0$ such that for all $\varepsilon > 0$, the union of $\varepsilon$-balls centred at spectra of $H^{(L)}$ can cover the interval $[\lambda_0(H^{(L)}),\lambda_0(H^{(L)})+c]$ for all sufficiently large $L$
	\footnote{Note that gapped is not the negation of gapless here. We adopt such strong definitions to exclude ambiguous cases such as systems with ground state degeneracies. Similar comments hold for order v.s. disorder and trivial v.s. critical.}.

	Our phase transition can also be detected by a local order parameter $\mathcal{O}=L^{-2}\sum_{i\in\Lambda}\mathcal{O}_i$, where each $\mathcal{O}_i$ is an on-site operator that is diagonal and has entries 0, $\pm 1$ in a computational basis. 
	In the gapless phase, 
	the ground state $\ket{\Psi}$ satisfies $\braket{\mathcal{O}}_\Psi=0$ and is hence \emph{disordered}; 
	while in the gapped phase it satisfies $\braket{\mathcal{O}}_\Psi=1$ and is hence \emph{ordered}.
	Moreover, the ground state in the gapless phase is \emph{critical}: 
	certain local operator shows algebraic decay of correlation;
	whereas the ground state in the gapped phase is \emph{trivial}: it is a product state.
	
	We note that everything in the constructed Hamiltonian is computable except for the tuning parameter $\phi$ which varies continuously. Indeed, they are either $\pi$ or algebraic numbers (rational or $\frac{\mathbb{Z}}{\sqrt{2}}$). 
	This is essential to ensure the non-triviality of our theorem. 
	Otherwise, consider the 1D quantum Ising model and shift the coupling parameter by an uncomputable number---the resulting Hamiltonian would be a trivial example for a phase transition at an uncomputable location.

	\section{Construction of the Hamiltonian}
	
	At a high level, our construction consists of two steps. 
	Loosely speaking, we first construct a Turing machine with oracle access to a real number $\phi$, such that it halts if and only if $\phi<\Omega_M$; 
	then we adopt a Quantum-Turing-Machine-to-Hamiltonian construction from \cite{BCW}, such that the halting and non-halting properties correspond to two distinct phases.
	The resulting Hamiltonian will have the desired phase diagram.

	\subsection{Chaitin's Constant and its Approximation}\label{sec-Chaitin}

	The Chaitin constant $\Omega_M$ of a Turing machine $M$ is a real number defined as follows:
	\begin{definition} \label{chaitin const}
		\begin{equation}\label{eq:chaitin}
			\Omega_M := \sum_{x: M(x) \text{~\upshape{halts}} } 2^{-|x|}.
		\end{equation}
	\end{definition}
	Here, each $x\in \{0,1\}^*$ is a finite bit string representing an input for the Turing machine $M$; 	$|x|$ is the length of $x$.
	The summation runs over all inputs such that $M$ halts.
	For technical reasons, it is assumed that $M$ is \emph{prefix-free}: if $M$ halts on both $x$ and $y$, then $x$ cannot be an initial segment of $y$ or vice versa.
	Under a suitable measure, $2^{-|x|}$ equals the probability for a random infinite bit string to start with $x$.
	Therefore, Chaitin's constant $\Omega_M$ may be interpreted as a halting probability: the probability for a prefix-free machine $M$ halts on a random input bit string. 
	Indeed, the series Eq.~(\ref{eq:chaitin}) converges due to the Kraft–McMillan inequality, and $\Omega_M\in(0,1)$.
	In the following, we sometimes omit the subscript $M$.

	The Chaitin's constant $\Omega$ for a universal Turing machine $M$ encodes the whole Halting problem---its first $n$ bits of $\Omega$ solve the Halting problem for all inputs of size $\leq n$. 
	It is therefore an uncomputable number \cite{chaitin}.
	Nevertheless, it can be computably approximated from below: there exists a computable sequence $\Omega_1\leq \Omega_2\leq \cdots$ such that $\lim_{s\to\infty}\Omega_s=\Omega$. 
	To do so, we simply simulate the computation for all inputs $x$ and add $2^{-|x|}$ to our approximation whenever we find $M(x)$ halts. 
	The idea is implemented in the the following algorithm that computes $\Omega_s$ from $s$.
	\begin{algorithm}[H]
		\caption{\textsc{Chaitin-approx\textsubscript{M}} 
		}
		\begin{algorithmic}[1]
			\State \textbf{input:} $s \in \mathbb{N}$
			\State set $\Omega_s=0$
			\For{$i=1$ to $s$}
			\State run $M$ for input $x_i$ for $s$ steps
			\If{$M(x_i)$ halts within $s$ steps}
			\State $\Omega_s \leftarrow 2^{-|x_i|}+\Omega_s$
			\EndIf
			\EndFor
			\State \textbf{output:} $\Omega_s$
		\end{algorithmic}
		\label{algo-chaitin-approx}
	\end{algorithm}
	Here we have fixed a way to list all bit strings $x\in\{0,1\}^*$. Since for each $s$ we only simulate $M$ on finitely many strings for finitely many steps, it is clear that $\Omega_s\leq \Omega_{s+1} < \Omega$ for all $s$.
	Moreover, for any $x$ such that $M(x)$ does halt, the computation $M(x)$ will be fully simulated for sufficient large $s$, hence 
	$\lim_{s\to\infty} \Omega_s = \Omega$.
	Note that there is no contradiction between the one-side approximability and the uncomputability. 
	Uncomputability of $\Omega$ only implies that no computable bounds on the quality of the approximation can exist, and that $\Omega$ cannot also be computably approximated from above.

	Now we describe an algorithm parametrized by a real number $\phi$ such that it halts if and only if $\phi<\Omega$.
	To do so, we enumerate the approximation-from-below sequence $\Omega_s$ and compare each $\Omega_s$ to $\phi$. 
	If we ever find $\phi<\Omega_s$ we halt; otherwise we continue.
	The idea is implemented in the algorithm \ref{witness-real}, denoted as as $W(\phi)$. 
	\begin{algorithm}[H] 
		\caption{\textsc{$W$ as a witness for $\phi<\Omega$} }
		\begin{algorithmic}[1]
			\State set $s=1$
			\While{$s>0$}
			\State $\Omega_s \gets \textsc{Chaitin-approx\textsubscript{M}}(s)$
			\If{$\phi < \Omega_s$}
			\State halt
			\EndIf
			\State $s \gets s+1$
			\EndWhile
		\end{algorithmic}
		\label{witness-real}
	\end{algorithm}
	Here, we have assumed oracle access to the full binary expansion of $\phi$.
	This enables the checking of whether $\phi<\Omega_s$ since the binary expansion of each $\Omega_s$ is finite and computable\footnote{As a convention, the binary expansion of $x=\frac{p}{2^q}$ where $p, q\in\mathbb{Z}$ is finite rather than ending in repeating 1s.}.
	However such an oracle is unrealistic and is given here for motivation only; we will modify $W$ to an algorithm that can actually be implemented in Sec.~\ref{sec-QTM}.

	The following theorem naturally follows from the construction, and shows that 	$W(\phi)$ has the desired ``phase transition" property: 
	as $\phi$ varies from $0$ to $1$, $W(\phi)$ switches from halting to non-halting exactly once, at $\Omega$. 
	\begin{theorem} \label{W halts}
		$W(\phi)$ halts in finite time if and only if $\phi < \Omega$.
	\end{theorem}
	\begin{proof}
		$W(\phi)$ halts if any only if there exists $s$ such that $\phi<\Omega_s$. Since the series $\Omega_s$ monotonically approximates $\Omega$, this happens if and only if $\phi < \Omega$.
	\end{proof}

	\subsection{The Quantum Turing Machine}\label{sec-QTM}
	Our next task is to construct a Hamiltonian whose low energy properties closely follows the computation $W$.
	To do so, we would like to implement $W$ as a quantum Turing machine (QTM), and then map it to a Hamiltonian via constructions in the field of Hamiltonian complexity.

	In $W$, we have assumed some access to $\phi$. 
	We will need to make this concrete.
	In our QTM, $\phi$ will be encoded into a unitary $U_\phi=\begin{pmatrix}
		1 & 0 \\
		0 & e^{i\pi\phi}
	\end{pmatrix}$, and will be extracted and digitized by quantum phase estimation (QPE).  	
	Note however that the QPE will not be exact for generic $\phi$ for two reasons. 
	First, the number of available qubits may not be enough to express $\phi$ in full (e.g., if $\phi$ is irrational); 
	second, some gates must be approximated since the inverse quantum Fourier transform in QPE requires length-dependent gates but our construction will require a finite gate set as we want our Hamiltonian to have bounded local dimension.
	Due to these QPE errors, the output quantum state will be a superposition of bit strings with amplitudes that are largest on bit strings that approximate $\phi$ well (see Appendix \ref{appendix:input extraction}) \cite{Cleve}. 
	
	To better control the amplitude on approximations to $\phi$ that are accurate to a desired precision, we will perform QPE to a large number of bits $n$ and then round every component in the superposition to a smaller number of bits $m$ using a rounding Turing machine $R$ (Appendix \ref{appendix:input extraction}). 
	The magnitude of $n$ will be constrained by the number of available qudits that the QPE QTM has access to and the precision $m$ will be $n$-dependent and will be specified later.
	
	Next, the output state is fed into a quantum Turing machine that executes the following classical algorithm, denoted by $W'$, a slight variation of $W$.
	The input is a number with a binary expansion of finite but unbounded length $n$, denoted as $0.\bar\phi$ where $\bar\phi\in\{0,1\}^n$, representing a possible approximation of the true value of $\phi\in[0,1)$; and the same-as-before integer $m$ that represents the number of bits that $0.\bar\phi$ has been rounded to.
	\begin{algorithm}[H]
		\caption{\textsc{$W'$ as a witness for $\phi<\Omega$} }
		\begin{algorithmic}[1]
			\State \textbf{input:} $\bar\phi \in \{0,1\}^n, m\in \mathbb{N}$
			\State $\Omega_m \gets \textsc{Chaitin-approx\textsubscript{M}}(m)$
			\If{$0<(0.\bar{\phi}) \upharpoonright m < \Omega_m\upharpoonright m$}
			\State halt
			\Else 
			\State loop forever
			\EndIf
		\end{algorithmic}
		\label{witness-rational}
	\end{algorithm}
	Here, $x\upharpoonright s$ denotes the truncation of $x\in(0,1)$ to the first $s$ bits of its binary expansion: $x\upharpoonright s=\frac{1}{2^s}\lfloor2^sx\rfloor$.
	The reason behind the truncation, as well as the looping in case of $(0.\bar{\phi}) \upharpoonright m=0$, will be clear in Sec. \ref{sec-error}.
	
	In short, our QTM is the dovetail of QPE, $R$, and $W'$.
	It is this three-step QTM that will be mapped into a Hamiltonian. 
	Since the QPE step requires the $U_\phi$ gate, the resulting Hamiltonian will also depend on $\phi$.

	\subsection{Error Analysis}\label{sec-error}
	
	It is crucial to make sure that $W'$ also exhibits a sharp ``transition" at $\phi=\Omega$ in a certain sense and that is not undermined by the QPE errors.
	
	As discussed above, the QPE will output a superposition of bit strings, each component of which will be rounded and inputted to $W'$,
	which effectively only reads $\tilde{\phi}_m=(0.\bar{\phi}) \upharpoonright m$.
	Here, we only consider the typical cases where $\tilde{\phi}_m$ is one of the best approximations of $\phi$: $\tilde{\phi}_m=\frac{1}{2^m}\lfloor2^m \phi \rfloor=\phi\upharpoonright m$ or $\tilde{\phi}_m=\frac{1}{2^m}\lceil2^m \phi \rceil$.
	The probability of atypical events (i.e. a poor approximation) will be bounded and be shown to have little contribution later.
	But even in typical cases, when comparing $\tilde{\phi}_m$ and $\phi$ we see that both underestimation and overestimation are possible.

	In the case of $\phi\in(0,\Omega)$, although $\tilde{\phi}_m$ may overestimate $\phi$ enough that $\tilde{\phi}_m>\Omega$, this will only be for small $m$. 
	For sufficiently large $m$, we will have $\tilde{\phi}_m\leq \phi+2^{-m}<\Omega$.
	And with a perhaps even larger $m$, we have $\tilde{\phi}_m<\Omega_m \upharpoonright m$ and hence $W'(\tilde{\phi}_m)$ will halt.
	In the case of $\phi\in[\Omega,1)$, although $\tilde{\phi}_m$ may underestimate $\phi$, 
	it always holds that $\tilde{\phi}_m\geq \phi\upharpoonright m\geq \Omega\upharpoonright m\geq \Omega_m \upharpoonright m$. 
	Therefore, $W'(\tilde{\phi}_m)$ will not halt for any $m$.

	In the case of $\phi \in [\Omega,1]$, for $\phi$ sufficiently close to $1$ from below, a good estimate to $\phi$ may be close to $0$ from above due to the periodicity of QPE.
	Such an estimate may cause $W'$ to incorrectly conclude that $\phi < \Omega$.
	This issue can be resolved by noticing that there is no $\Omega$ for which we expect $W'(1)$ to halt as $\Omega\leq 1$. 
	Therefore we can require $W'$ to effectively interpret $0^m$ as $1$ by manually setting it to loop indefinitely for such an input.
	And so the periodicity cannot cause $QPE+R+W'$ to halt for $\phi$ sufficiently close to $1$.
	
	Overall therefore, when given a good estimate of $\phi$ as input, $W'$ will never halt for $\phi\in[\Omega,1]$, whereas for $\phi\in(0,\Omega)$, $W'$ will halt if given enough workspace.
	
	The above observations are summarized in the following theorem (see Appendix \ref{sec:thm:W'transition} for a detailed proof).
	\begin{theorem}\label{thm: W'transition}
		Fix $\phi\in(0,1]$ and consider sequences $\{\tilde{\phi}_m\}$, where $\tilde{\phi}_m=\frac{1}{2^m}\lfloor2^m \phi \rfloor$ or $\frac{1}{2^m}\lceil 2^m \phi \rceil$. We have:
		\begin{itemize}
			\item if $\phi<\Omega$, then for large enough $m$,  $0<\tilde{\phi}_m<\Omega_m\upharpoonright m$ and hence $W'(\tilde{\phi}_m)$ halts;
			\item if $\phi\geq\Omega$, then for all $m$, $\tilde{\phi}_m=0$ or $\tilde{\phi}_m>\Omega_m\upharpoonright m$ and hence $W'(\tilde{\phi}_m)$ does not halt.
		\end{itemize}
	\end{theorem}

	\subsection{From QTM to Hamiltonian}\label{sec-QTMtoH}

	The celebrated Feynman-Kitaev clock construction \cite{kitaev} provides a method for encoding quantum circuits into the ground state of a Hamiltonian. Later, Gottesman and Irani \cite{GI09} extended it to quantum Turing machines.
	Given a QTM, the resulting Hamiltonian $H_{\text{QTM}}$ will be one-dimensional and translationally invariant with nearest-neighbour interactions. 
	On a chain of $L$ qudits, the unique ground state of $H_{\text{QTM}}$ will take the form of a \emph{history state}:
	\begin{equation}
		\frac{1}{\sqrt{T+1}} \sum_{t=0}^T \ket{t} U_t \dots U_1 \ket{\psi},
	\end{equation}
	encoding a computation of a QTM on a tape of length $L-2$ for time $T$.
	In our construction, $T=T(L)=\exp(\Theta(L))$ is determined by $L$ and does not depend on the computation (e.g., whether it halts or not) \cite{CPW},
	$\{\ket{t}\}$ are orthonormal states of the clock qudits.

	We can further make the ground state energy dependent on chosen aspects of the computation by adding certain penalty terms (positive energy), resulting in a Hamiltonian, denoted $H_{\text{comp}}$, encoding QPE followed by $R$ followed by $W'$.
	In our construction, we will penalize non-halting cases.
	With a refined analysis of the ground state energy of the clock construction, we find (see Appendix \ref{app:FKgse} theorem \ref{thm:tightgap} and Appendix \ref{appendix: hamiltonian encoding} theorem \ref{L asympt})
	\begin{equation}
		\lambda_0(H_{\comp}(\phi))\sim\frac{p_{\text{\upshape{nh}}}(L)}{T^2},
	\end{equation}
	where $p_{\text{\upshape{nh}}}(L)$ is the probability that QPE outputs a finite string for which $W'$ does not halt within time $T(L)$.

	For a given $\phi$, due to Theorem \ref{thm: W'transition} and the fact that the output of QPE is concentrated around good approximations to $\phi$, we can conclude that the ground state energy of $H_{\comp}(\phi)$ is dependent on whether $\phi$ or $\Omega$ is larger:
	\begin{equation}\label{eq-Ecomp}
		\lambda_0(H_{\comp}(\phi))
		=\begin{cases}
			o(\frac{1}{T^2}) \text{~for large enough $L$},~~&\text{if $\phi<\Omega$ (halting)},\\
			\Theta(\frac{1}{T^2}) \text{~for all $ L$},~~&\text{if $\phi\geq\Omega$ (non-halting)}.
		\end{cases}
	\end{equation}
	That is, the two energies differ asymptotically. 
	Despite that QPE is both approximate and not generally deterministic, 
	our chosen classical algorithm $W'$ guarantees that the desired feature---the sharp ``transition" at $\phi=\Omega$--- carries over into the ground state energy of $H_{\comp}(\phi)$.

	\subsection{Constant Spectral Gap Separation}
	Although $\lambda_0(H_\text{comp})$ is dependent on whether $\phi<\Omega$, the energy difference between the two cases vanishes in the thermodynamic limit.
	In the following, we will leverage the (small) difference in the ground state energy to introduce a constant spectral gap separation between the two cases.
	This is achieved by the tiling\&marker construction of \cite{BCW} followed by a manipulation from \cite{Bausch_2020} (see Appendix \ref{appendix: hamiltonian encoding} for full details).

	In short, we design a classical tiling Hamiltonian so that each zero-energy ground state corresponds to a periodic partition of the lattice into squares of equal size $s$.
	The ground state subspace of the tiling Hamiltonian is highly degenerate and $s$ is unfixed. 
	The tiling Hamiltonian is then coupled to a quantum marker Hamiltonian \cite{Bausch_2020}, which can effectively read off the value of $s$ and introduce a small bonus (negative energy) inversely dependent on $s$.
	The tiling\&marker Hamiltonian is then coupled to $H_\text{comp}$, and the resulting combined Hamiltonian is denoted as $H_\uu$:
	\begin{equation}
		H_\uu=H_\text{comp}+H_\text{tiling\&marker}+\text{couplings}.
	\end{equation}
	In the ground state of the $H_\uu$, ground states of $H_\text{comp}$ (the history states) are placed on the top of each (size-$s$) square, and the ground state of the $H_\uu$ is the product state over squares. 
	The energy for each square is the sum of the contributions from $H_\text{comp}$ (Eq.(\ref{eq-Ecomp}), positive) and the marker Hamiltonian (negative).
	
	Importantly, with a well-chosen marker Hamiltonian, the (absolute value of) energy bonus can lie exactly in between the two cases in Eq. (\ref{eq-Ecomp}), so that the ground state energy of $H_\uu$ on each square will be negative and positive in the halting and non-halting case, respectively.
	
	If $\phi < \Omega$, we can consider an $s$ that is large enough for $W'$ to halt. 
	Such configuration will have a negative energy \emph{density}, as each of the size-$s$ squares contributes a small but non-zero negative energy. 
	The ground state energy hence diverges to $-\infty$ as $L\to\infty$.
	On the other hand, if $\phi \geq \Omega$ then $W'$ will never halt on $\phi$ and so the energy is always positive no matter how large $s$ is. 
	In this case, taking $s$ as large as the full system size, we see that the ground state energy vanishes from above as $L\to\infty$. 
	Therefore, we have constructed a Hamiltonian with the following behaviour in the thermodynamic limit:
	\begin{equation}
		\lambda_0(H_\uu(\phi))~
		\begin{cases}
			
			=-\infty, ~~&\text{if $\phi<\Omega$ (halting)},\\
			
			\geq 0,~~&\text{if $\phi\geq\Omega$ (non-halting)}.
			
		\end{cases}
	\end{equation} 
	
	Lastly, by a standard manipulation consisting of adding extra Hamiltonians with known spectra in a specific way (see Appendix \ref{appendix: hamiltonian encoding}) \cite{Bausch_2020}, we can construct a Hamiltonian $H_\tot(\phi)$ out of $H_\uu(\phi)$, such that $H_\tot(\phi)$ is gapless (or gapped) if $\lambda_0(H_\uu(\phi))=-\infty$ (or $\geq 0$).
	The Hamiltonian $H_\tot(\phi)$ is the final construction that we presented in Theorem \ref{theorem 1}.

	In fact, the last step is quite flexible and we can arrange the ground state of $H(\phi)$ to contain the ground state of the decoupled 1D XY spin chains if $\phi < \Omega$, and be the product state $\ket{0}^{\otimes \Lambda}$ if $\phi \geq \Omega$.
	Hence our phases are also discriminated by critical v.s trivial.
	Moreover, we can define $\mathcal{O}=L^{-2}\sum_{i\in\Lambda}\mathcal{O}_i$ where each local operator $\mathcal{O}_i$ acts as identity on $\ket{0}_i$ and annihilates the Hilbert space that the XY spin chains live in.
	Therefore, $\braket{\mathcal{O}}=0$ in the gapless phase and $\braket{\mathcal{O}}=1$ in the gapped phase.

	\section{Discussion}
	We have explicitly constructed a family of two-dimensional, translationally-invariant, nearest-neighbour Hamiltonians continuously parametrized by a real parameter $\phi$.	
	Restricting to the unit interval $\phi\in(0,1]$, we proved that there are exactly two extended phases spanning the entire parameter space, separated by a single phase transition.
	The phase transition point corresponds to Chaitin's constant $\Omega$, an intriguing real number that is well-defined but uncomputable. 
	Our results show that even in the very simple setting of only allowing a single phase transition, 
	there is still a fundamental complexity barrier to precisely determining the transition point via numerical simulations, finite-sized experiments, or mathematical reasoning.

	Uncomputable phase diagrams have have been explored in previous work \cite{BCW}, where the phase diagram maps out all instances of the Halting problem sequentially on the real line. 
	Regions corresponding to halting/non-halting are distributed in a random, uncomputable manner and alternate infinitely.
	In contrast, our construction distils this uncomputability to a single uncomputable point that encapsulates all instances of the Halting problem.
	This approach provides better analytical control over the phase diagram, allowing us to understand every point of the phase diagram modulo the intrinsic complexity of the Chaitin's constant (for which many properties are also known).
	In contrast, in \cite{BCW}, regions that can be shown to have definite gap/gapless properties are disconnected intervals with rational endpoints and do not encompass the entire phase diagram.

	Our phase diagram is periodic due to the modular arithmetic inherent in quantum phase estimation. It is straightforward to construct a model with a single phase transition over the entire real axis by substituting $\phi$ with $(1+\exp(-\phi))^{-1}$.
	Furthermore, while two phases are discriminated by gapped v.s. gapless, critical v.s. trivial, and order v.s. disorder in this work, other distinguishing properties can also be incorporated.
	By modifying the final step of our construction, we can ensure that the model has two phases: one exhibiting a desired property and the other lacking it \cite{CPW}.

	Although $\Omega$ is highly random in the precise technical sense of Martin-L\"of randomness \cite{Downey_Hirschfeldt_2010}, this property is not necessary for the uncomputability of the phase diagram. 
	Our construction would, for example, allow us to dilute the randomness in $\Omega$ by inserting a 0 between every bit in $\Omega$, and the resulting phase transition point would remain uncomputable.
	Our construction does crucially rely on the one-sided approximability (i.e. that $\Omega$ is left-computably enumerable, or left-c.e.).
	It is known that left-c.e. numbers are exactly the halting probabilities of prefix-free Turing machines.
	To construct an interesting phase diagram, we require it to be uncomputable, which is guaranteed by the universality of the Turing machine \footnote{Universality is not necessary though \cite{downey2000presentations}.}.

	Our results illustrate uncomputable numbers may \emph{emerge} as phase transition points in physics-like models, even when all underlying, microscopic data are fully computable.
	While our construction works for all one-sided computably enumerable numbers, one interesting question is whether there are other types of uncomputable numbers that could appear as a phase transition point.
	Another intriguing open question is whether and how uncomputable numbers may emerge in other physical contexts.

	\section*{Acknowledgment}
	
	J.P. is supported by the Engineering and Physical Sciences Research Council (grant number EP/S021582/1). Z.L. is supported by Perimeter Institute; research at Perimeter Institute is supported in part by the Government of Canada through the Department of Innovation, Science and Economic Development and by the Province of Ontario through the Ministry of Colleges and Universities. This work has been supported in part by the EPSRC Prosperity Partnership in Quantum Software for Simulation and Modelling (grant EP/S005021/1), and by the UK Hub in Quantum Computing and Simulation, part of the UK National Quantum Technologies Programme with funding from UKRI EPSRC (grant EP/T001062/1).

	\bibliographystyle{alpha}
	\bibliography{ref}

\newcommand{\etalchar}[1]{$^{#1}$}
\begin{thebibliography}{BCLPG20}

\bibitem[AGIK09]{aharonov2009power}
Dorit Aharonov, Daniel Gottesman, Sandy Irani, and Julia Kempe.
\newblock The power of quantum systems on a line.
\newblock {\em Communications in mathematical physics}, 287(1):41--65, 2009.

\bibitem[AKLT87]{AKLT}
Ian Affleck, Tom Kennedy, Elliott~H. Lieb, and Hal Tasaki.
\newblock Rigorous results on valence-bond ground states in antiferromagnets.
\newblock {\em Phys. Rev. Lett.}, 59:799--802, Aug 1987.

\bibitem[BCLPG20]{Bausch_2020}
Johannes Bausch, Toby~S Cubitt, Angelo Lucia, and David Perez-Garcia.
\newblock Undecidability of the spectral gap in one dimension.
\newblock {\em Physical Review X}, 10(3):031038, 2020.

\bibitem[BCW21]{BCW}
Johannes Bausch, Toby~S Cubitt, and James~D Watson.
\newblock Uncomputability of phase diagrams.
\newblock {\em Nature Communications}, 12(1):452, 2021.

\bibitem[Ben73]{bennett1973logical}
Charles~H Bennett.
\newblock Logical reversibility of computation.
\newblock {\em IBM journal of Research and Development}, 17(6):525--532, 1973.

\bibitem[Ben89]{bennett1989time}
Charles~H Bennett.
\newblock Time/space trade-offs for reversible computation.
\newblock {\em SIAM Journal on Computing}, 18(4):766--776, 1989.

\bibitem[BG15]{bravyi-gosset}
Sergey Bravyi and David Gosset.
\newblock {Gapped and gapless phases of frustration-free spin-1/2 chains}.
\newblock {\em Journal of Mathematical Physics}, 56(6):061902, 06 2015.

\bibitem[CEMM98]{Cleve}
Richard Cleve, Artur Ekert, Chiara Macchiavello, and Michele Mosca.
\newblock Quantum algorithms revisited.
\newblock {\em Proceedings of the Royal Society of London. Series A:
  Mathematical, Physical and Engineering Sciences}, 454(1969):339--354, 1998.

\bibitem[Cha75]{chaitin}
Gregory~J Chaitin.
\newblock A theory of program size formally identical to information theory.
\newblock {\em Journal of the ACM (JACM)}, 22(3):329--340, 1975.

\bibitem[Cha92]{chaitin1992information}
Gregory~J Chaitin.
\newblock Information-theoretic incompleteness.
\newblock {\em Applied Mathematics and Computation}, 52(1):83--101, 1992.

\bibitem[CPGW15]{CPW}
Toby~S Cubitt, David Perez-Garcia, and Michael~M Wolf.
\newblock Undecidability of the spectral gap.
\newblock {\em Nature}, 528(7581):207--211, 2015.

\bibitem[CPGW22]{CPWfull}
Toby Cubitt, David Perez-Garcia, and Michael~M Wolf.
\newblock Undecidability of the spectral gap.
\newblock In {\em Forum of Mathematics, Pi}, volume~10, page e14. Cambridge
  University Press, 2022.

\bibitem[DH10]{Downey_Hirschfeldt_2010}
Rodney~G Downey and Denis~R Hirschfeldt.
\newblock {\em Algorithmic randomness and complexity}.
\newblock Springer Science \& Business Media, 2010.

\bibitem[DL00]{downey2000presentations}
Rodney~G Downey and Geoffrey~L LaForte.
\newblock Presentations of computably enumerable reals.
\newblock Technical report, Department of Computer Science, The University of
  Auckland, New Zealand, 2000.

\bibitem[FNW92]{fannes1992finitely}
Mark Fannes, Bruno Nachtergaele, and Reinhard~F Werner.
\newblock Finitely correlated states on quantum spin chains.
\newblock {\em Communications in mathematical physics}, 144:443--490, 1992.

\bibitem[GI09]{GI09}
Daniel Gottesman and Sandy Irani.
\newblock The quantum and classical complexity of translationally invariant
  tiling and hamiltonian problems.
\newblock In {\em 2009 50th Annual IEEE Symposium on Foundations of Computer
  Science}, pages 95--104. IEEE, 2009.

\bibitem[GM16]{gosset-mozgunov}
David Gosset and Evgeny Mozgunov.
\newblock {Local gap threshold for frustration-free spin systems}.
\newblock {\em Journal of Mathematical Physics}, 57(9):091901, 09 2016.

\bibitem[Hei28]{Heisenberg1928}
W.~Heisenberg.
\newblock Zur theorie des ferromagnetismus.
\newblock {\em Zeitschrift für Physik}, 49(9-10):619--636, 1928.

\bibitem[Jor75]{jordan1875essai}
Camille Jordan.
\newblock Essai sur la g{\'e}om{\'e}trie {\`a} $ n $ dimensions.
\newblock {\em Bulletin de la Soci{\'e}t{\'e} math{\'e}matique de France},
  3:103--174, 1875.

\bibitem[Kit03]{kitaev2003fault}
Alexei Kitaev.
\newblock Fault-tolerant quantum computation by anyons.
\newblock {\em Annals of physics}, 303(1):2--30, 2003.

\bibitem[Kit06]{kitaev2006anyons}
Alexei Kitaev.
\newblock Anyons in an exactly solved model and beyond.
\newblock {\em Annals of Physics}, 321(1):2--111, 2006.

\bibitem[KKR06]{kempe2006complexity}
Julia Kempe, Alexei Kitaev, and Oded Regev.
\newblock The complexity of the local hamiltonian problem.
\newblock {\em Siam journal on computing}, 35(5):1070--1097, 2006.

\bibitem[Kna88]{knabe1988energy}
Stefan Knabe.
\newblock Energy gaps and elementary excitations for certain vbs-quantum
  antiferromagnets.
\newblock {\em Journal of statistical physics}, 52:627--638, 1988.

\bibitem[KSV02]{kitaev}
Alexei Kitaev, Alexander Shen, and Mikhail~N Vyalyi.
\newblock {\em Classical and quantum computation}.
\newblock Number~47. American Mathematical Soc., 2002.

\bibitem[LM19]{lemm-mozgunov}
Marius Lemm and Evgeny Mozgunov.
\newblock {Spectral gaps of frustration-free spin systems with boundary}.
\newblock {\em Journal of Mathematical Physics}, 60(5):051901, 05 2019.

\bibitem[LSM61]{lieb1961two}
Elliott Lieb, Theodore Schultz, and Daniel Mattis.
\newblock Two soluble models of an antiferromagnetic chain.
\newblock {\em Annals of Physics}, 16(3):407--466, 1961.

\bibitem[MRR{\etalchar{+}}53]{metropolis1953equation}
Nicholas Metropolis, Arianna~W Rosenbluth, Marshall~N Rosenbluth, Augusta~H
  Teller, and Edward Teller.
\newblock Equation of state calculations by fast computing machines.
\newblock {\em The journal of chemical physics}, 21(6):1087--1092, 1953.

\bibitem[OT05]{oliveira2005complexity}
Roberto Oliveira and Barbara~M Terhal.
\newblock The complexity of quantum spin systems on a two-dimensional square
  lattice.
\newblock {\em arXiv preprint quant-ph/0504050}, 2005.

\bibitem[Sol99]{solovay1999version}
Robert~M Solovay.
\newblock A version of {$\Omega$} for which {ZFC} can not predict a single bit.
\newblock Technical report, Department of Computer Science, The University of
  Auckland, New Zealand, 1999.

\bibitem[Tur36]{turing1936computable}
Alan~Mathison Turing.
\newblock On computable numbers, with an application to the
  entscheidungsproblem.
\newblock {\em J. of Math}, 58(345-363):5, 1936.

\bibitem[Wat19]{Watson19}
James~D Watson.
\newblock Detailed analysis of circuit-to-hamiltonian mappings.
\newblock {\em arXiv preprint arXiv:1910.01481}, 2019.

\bibitem[Whi92]{dmrg}
Steven~R White.
\newblock Density matrix formulation for quantum renormalization groups.
\newblock {\em Phys. Rev. Lett.}, 69:2863--2866, Nov 1992.

\end{thebibliography}
	\appendix
	
	\section{Quantum Turing Machines: Phase Estimation, Rounding, Witness}\label{app:3QTM}
	\subsection{Input Extraction via Phase Estimation and Rounding} \label{appendix:input extraction}
	Given the unitary gate $U_\phi=\begin{pmatrix}
		1 & 0 \\
		0 & e^{i\pi\phi}
	\end{pmatrix}$, quantum phase estimation (QPE) tries to print the binary expansion of $\phi$.
	A generic real number has an infinitely long binary expansion and so infinitely many qudits are required to represent it exactly. 
	However, we want to avoid necessitating infinitely many qudits and only run the QPE on finite, although arbitrarily large, systems. 
	The requirement that the construction should work with finite tape is common in Hamiltonian computability constructions, in order to overcome a vanishing energy gap (Eq.~(\ref{eq-Ecomp}).
	
	Fixing $\phi$, our goal is to output the best $m$-digit approximation of it, $\frac{1}{2^m} \floor{2^m \phi}$ or $\frac{1}{2^m} \ceil{2^m \phi}$,  with high probability.
	It is known that QPE with finitely many qudits outputs the best approximation to $\phi$ with probability at least $4/\pi^2$ \cite{Cleve}. 
	However, this is insufficient for our purpose. 
	To further improve the probability, we will run QPE with more qudits and then round the output to a smaller number of bits.
	
	\begin{lemma}\label{lemma:QPEn-m}
		For any $0 < m < n$, with probability $\geq 1-2^{-(n-m)}$ phase estimation on  $\phi \in [0,1)$ to $n$ bits of precision yields an $n$-bit estimate $\phi'$ satisfying $|\phi'-\phi|<\frac{1}{2^{-(m+1)}}$.
	\end{lemma}
	
	\begin{proof} (see also \cite{Cleve})
		Denote $\phi=\phi_0+\frac{p}{2^n}$ where $(p\in\mathbb{Z}, 0\leq \phi_0<\frac{1}{2^n})$.
		If $\phi_0=0$, then QPE is already exact. 
		If $\phi_0\neq 0$, we define $\delta(z):= \phi - \frac{z}{2^n}$ for each $z\in\mathbb{Z}$. 
		Due to the periodicity, we can demand $|\delta(z)|\leq\frac{1}{2}$.
		Then for any $z$,
		\begin{equation}
			Pr[z] = \frac{\sin^2(\pi2^n\delta(z))}{2^{2L}\sin^2(\pi\delta(z))} \leq \frac{\sin^2(\pi 2^n\phi_0)}{4^{L+1}\delta(z)^2}.
		\end{equation}
		Summing over terms with large deviations,
		\begin{equation}
			\begin{aligned}
				&Pr\left[z : |\delta(z)| \geq \frac{1}{2^{m+1}} \right]\\ 
				\leq&\frac{\sin^2(\pi 2^n\phi_0)}{4^{L+1}}
				\left[\sum_{k=0}^{\infty}\frac{1}{\left( \frac{1}{2^{m+1}} + \phi_0 +\frac{k}{2^n} \right)^2} 
				+\sum_{k=0}^{\infty}\frac{1}{\left( \frac{1}{2^{m+1}} +\frac{1}{2^n}-\phi_0 +\frac{k}{2^n} \right)^2} \right]		
				\\
				\leq& \frac{\sin^2(\pi 2^n\phi_0)}{2}
				\left(\frac{1}{2^{n-m}+2^{n+1}\phi_0-1}+
				\frac{1}{2^{n-m}-2^{n+1}\phi_0+1}
				\right)\\
				\leq&\frac{1}{2^{n-m}}.
			\end{aligned}
		\end{equation}
		The third line comes from $\frac{1}{x^2}<\frac{1}{2}(\frac{1}{x-1}+\frac{1}{x+1})$. 
		The last line comes from $\cos ^2\left(\frac{\pi  x}{2}\right)<1-\frac{x^2}{4}$ when $x\in(-1,1)$.
	\end{proof}

	This is still not yet sufficient, since an estimate that is within $\frac{1}{2^{m+1}}$ of $\phi$ need not necessarily have the same first $m$ bits as $\phi$ (for example, the first $m$ bits might be $\phi \upharpoonright m - 2^{-m}$). 
	To solve this problem, we define a ``rounding" function $r_m$ as follows:
	\begin{equation}
		r_m(x)=\begin{cases}
			x+\frac{1}{2^{m}},&~~\text{if the $(m+1)^{\text{th}}$ bit is 1}\\
			x,&~~\text{otherwise}
		\end{cases}.
	\end{equation}
	Although this is not standard rounding, it has the advantage of being trivial to construct a reversible Turing machine $R$ that implements the rounding within $n$ space.

	We define a set
	\begin{equation}\label{eq:defIm}
		\mathcal{I}_m(x)=\left\{\frac{1}{2^m} \floor{2^m x},\frac{1}{2^m} \ceil{2^m x}\right\}.
	\end{equation} 
	If $2^mx\in\mathbb{N}$ then $\mathcal{I}_m(x)=\{x\}$.
	All computations here are understood as modulo $1$.
	
	\begin{lemma}\label{lemma:QPErounding}
		For any $n \in \mathbb{N}$ and any $0 < m \leq n$, running QPE on $\phi \in [0,1)$ to $n$ bits of precision and then dovetailing with the rounding TM $R$ on input $m$ yields an output state $\ket{\chi}$ such that
		\begin{equation}
			\text{\upshape{tr}}(\ket{\chi}\bra{\chi} \cdot \mathcal{P}_{\mathcal{I}_m(\phi)}) \geq 1 - 2^{-(n-m)},
		\end{equation}
		where $\mathcal{P}_{\mathcal{I}_m(\phi)}$ projects onto $\text{\upshape{span}}(\ket{\tilde{\phi}_m})$ for $\tilde{\phi}_m \in \mathcal{I}_m(\phi)$. 
	\end{lemma}
	\begin{proof} 
		By lemma \ref{lemma:QPEn-m}, it is enough to show that, if $|\phi'-\phi| < \frac{1}{2^{m+1}}$ then $r_m(\phi')\upharpoonright m\in\mathcal{I}_m(\phi)$.
		First assume $2^m\phi\notin\mathbb{N}$. 
		Since $\phi\in[\frac{1}{2^m}\floor{2^m\phi},\frac{1}{2^m}\ceil{2^m\phi}]$, we have 
		\begin{equation}
			\phi'\in(x_1,x_2)\cup[x_2,x_3)\cup[x_3,x_4)\cup[x_4,x_5),
		\end{equation}
		where
		$x_1=\frac{1}{2^m}\floor{2^m\phi}-\frac{1}{2^{m+1}}$, 
		$x_2=\frac{1}{2^m}\floor{2^m\phi}$,
		$x_3=\frac{1}{2^m}\floor{2^m\phi}+\frac{1}{2^{m+1}}$,
		$x_4=\frac{1}{2^m}\ceil{2^m\phi}$,
		$x_5=\frac{1}{2^m}\ceil{2^m\phi}+\frac{1}{2^{m+1}}$.
		We know $r_m(\phi')=\phi'+\frac{1}{2^{m}}$ for $\phi'$ in the first and third interval, while $r_m(\phi')=\phi'$ otherwise.
		Hence,
		\begin{equation}
			r_m(\phi')\in\left[\frac{1}{2^m}\floor{2^m\phi},\frac{1}{2^m}\ceil{2^m\phi}+\frac{1}{2^m}\right).
		\end{equation}
		It is then clear that $r_m(\phi')\upharpoonright m\in\mathcal{I}_m(\phi)$.
		
		The case where $2^m\phi\in\mathbb{N}$ is even simpler.
	\end{proof}

	The above lemma works for QPE with exact gates depending on $n$.
	In our construction, we need to find (for each $\phi$) a single quantum Turing machine such that when implemented on tapes of varying lengths, it performs QPE for different values of $n$.
	The QPE-QTM does not know $n$ a priori and must be $n$-independent.
	Therefore, we have to use a fixed gate set to synthesize the $n$-dependent gate sets.
	This will result in another source of error to be analysed.
	
	\begin{lemma}\label{QTM family}
		There exists a family of QTMs $P_\phi$ indexed by $\phi \in [0,1)$, all with identical internal states and symbols but differing transition rules, with the property that on input $L$ in unary, $P_\phi$ halts after $O(\text{\upshape{\text{poly}}}(n)2^n)$ steps, uses $n+3$ tape, and outputs $\ket{\tilde{\chi}}$ s.t. for all $0<m\leq n$, after dovetailing with the rounding TM $R$ on input $m$ in unary we have
		\begin{equation}
			\text{\upshape{tr}} \left( \ket{\tilde{\chi}}\bra{\tilde{\chi}} \cdot 
			\mathcal{P}_{\mathcal I_m(\phi)} \right) \geq 1-2^{-(n-m)} -\delta(n),
		\end{equation}
		where $\delta(n) < \left(\frac{n^2}{2}\right)2^{-c_2n^{1/c_1}}$ for constants $3 < c_1 < 4$, $c_2\geq 1$ that can be written down explicitly.
	\end{lemma}
	\begin{proof}
		Take the family of QPE QTMs from \cite{CPW} and modify such that all rotation gates are Solovay-Kitaev approximated (as in \cite[Lem. B.6]{BCW}). This gives $\delta(n)$ as claimed ($\frac{n^2}{2}$ gates need to be approximated, and within space $n$ the precision of each is limited to $2^{-c_2n^{1/c_1}}$). Combined with lemma \ref{lemma:QPErounding}, we are done.
	\end{proof}

	\subsection{Error Analysis: More on Theorem \ref{thm: W'transition}}\label{sec:thm:W'transition}

	The algorithm $W'$ will effectively read $\tilde{\phi}_m=(r_m{\phi'}) \upharpoonright m$.
	Typically, QPE may have underestimated $\phi$ as $\frac{1}{2^m}\lfloor2^m\phi\rfloor$, which equals $\phi \upharpoonright m$; or overestimated $\phi$ as $\frac{1}{2^m}\lceil2^m\phi\rceil$, which equals $\phi \upharpoonright m+2^{-m}$ if $2^m\notin\mathbb{N}$.
	Here we prove Theorem \ref{thm: W'transition}, which claims that for underestimation and overestimation within these ranges, a sharp halting/non-halting separation is retained at $\Omega$.

	First, consider the case $\phi\in(0,\Omega)$.
	\begin{lemma}\label{preserves}
		If $\phi < \Omega$ then there exists $m_0<\infty$ such that $\phi+2^{-m} < \Omega_m \upharpoonright m$ for all $ m>m_0$.
	\end{lemma}
	\begin{proof}
		We claim that $\lim_{s\to\infty}\Omega_s\upharpoonright s=\Omega$.
		Indeed, fixing a $m$, since $\Omega_s$ converge to $\Omega$ from below, we have $\Omega_s \upharpoonright m=\Omega \upharpoonright m$ for large enough $s$. 
		Therefore $\lim_{s\to\infty}\Omega_s\upharpoonright s\geq \Omega \upharpoonright m$.
		Since $m$ is arbitrary, we get $\lim_{s\to\infty}\Omega_s\upharpoonright s\geq \Omega$. The other direction is obvious.
		
		For $\phi < \Omega$, we can pick $m_1$ such that $\phi+2^{-m_1}<\Omega$. 
		We pick $m_2$, such that $\phi+2^{-m_1}<\Omega_{m_2} \upharpoonright m_2$.
		Now, pick $m_0=\max\{m_1,m_2\}$, the lemma is proved.
	\end{proof}
	No matter whether $\tilde{\phi}_m=\frac{1}{2^m}\lfloor2^m\phi\rfloor$ or $\frac{1}{2^m}\lceil2^m\phi\rceil$, it always satisfies $\tilde{\phi}_m<\phi+2^{-m}$. 
	Moreover, for a fixed $\phi\in(0,\Omega)$, $\frac{1}{2^m}\lfloor2^m\phi\rfloor\neq 0$ for sufficiently large $m$.
	This fact, combined with the above lemma, implies that $0<\tilde{\phi}_m<\Omega_m \upharpoonright m$ for sufficiently large $m$, implying that $W'(\tilde{\phi}_m)$ will halt for sufficiently large $m$.

	Next, consider the case $\phi\in[\Omega,1]$.
	No matter whether $\tilde{\phi}_m=\frac{1}{2^m}\lfloor2^m\phi\rfloor$ or $\frac{1}{2^m}\lceil2^m\phi\rceil$,
	we always have  $\tilde{\phi}_m\geq \phi\upharpoonright m\geq \Omega\upharpoonright m\geq \Omega_m \upharpoonright m$ for all $m$ due to the monotonicity of $\Omega_m$ and the truncation.
	
	Even if in the special case such that QPE may overestimate $\phi$ as 1 and returns all zeros on the $m$ bits (since QPE works modulo 1), it is still a non-halting case: we have manually defined $W'$ to loop forever upon reading $0^m$.
	
	Therefore, $W'(\tilde{\phi}_m)$ does not halt for any $m$ if $\phi\in[\Omega,1]$.

	\subsection{Constructing the Quantum Turing Machine}
	
	We have almost finished the construction of a three-stage QTM consisting of the dovetail of QPE, $R$, and $W'$. 
	The QPE QTM approximately extracts a parameter $\phi$ using phase estimation and prints it on a tape in superposition. 
	For each element of the superposition, $R$ reversibly rounds it to $m$ bits of precision (morally). 
	Then $W'$ generates an approximation to Chaitin's constant, and compares it with the output of $R$.
	
	There are two strict requirements this three-stage QTM must satisfy.
	Firstly, for the three-stage QTM to be a \emph{quantum} TM it must be reversible (unitary). 
	Secondly, we will soon encode the three-stage QTM in the low energy subspace of a Hamiltonian on qudit chains of arbitrarily long but finite length. 
	The length of the qudit chain $L$ corresponds to the QTM tape length $n$ up to a constant.
	We will require that the QTM only requires constant space overhead---i.e. that it can extract $n$-bit estimates of $\phi$ and perform calculations on them using only $n+constant$ workspace.
	This $n+constant$ workspace can easily then be made to fit inside a qudit chain of length $L$.
	This second requirement makes satisfying reversibility more challenging. 
	It has been known how to make any Turing machine reversible since the 1970s \cite{bennett1973logical}, but the technique requires $\log T$ overhead \cite{bennett1989time}.
	Fortunately, \cite{CPWfull} explicitly constructs a toolbox of useful space-efficient QTMs that we shall exploit.

	As stated in lemma \ref{QTM family}, the QPE QTM is reversible and requires at most $n+3$ workspace.
	The classical ``rounding'' TM $R$ is also reversible and requires no more than a length of $n+3$. 
	To construct the space-bounded three-stage QTM it remains to show that $W'$ can also be made reversible and space-efficient; i.e. that the approximation to $\Omega$ can be generated reversibly with no greater than constant space overhead, and likewise for the subsequent comparison to the approximation to $\phi$.

	We shall first analyse the generation of approximations to $\Omega$ (alg. \ref{algo-chaitin-approx}). 
	Order $x \in \{0,1\}^*$ first by length and then lexicographically. 
	On $L$ tape, for a given input $m\leq L$, over the course of the algorithm the bit strings $x_1, x_2, \dots, x_m$ will be written on the tape. 
	For efficiency of space, we shall set the algorithm to increment $x_i$ to $x_{i+1}$ once $M(x_i)$ has been run for $m$ steps, instead of writing out $x_{i+1}$ on a new section of the tape. 
	As $\max |x_i| = |x_m| \leq L$, there will always be enough space to write each string. 
	For each $x_i$, $M(x_i)$ is run for exactly $m$ steps. 
	The Turing machine $M$ must be reversible for this to be reversible. 
	As within $m$ time steps, the most space that can be used is $m$ cells, this procedure is also tightly space bounded.
	Finally, we must also add $2^{-|x_i|}$ to a number on a different tape (that will be output as $\Omega_m$ at the algorithm end) if $M(x_i)$ has halted after $m$ steps. 
	This can be done reversibly and space-efficiently using a QTM from the \cite{CPWfull} toolbox.
	
	For the latter subroutine, the QTMs from \cite{CPWfull} include a checker for equality between two bit strings.
	This equality-checking QTM can be trivially modified to include an extra argument $m$, where $m$ is a number written in unary, such that the new \textsc{Equality2} TM only tests for equality on the first $m$ bits of the two inputs. 
	Then by using \textsc{Equality2} we can construct a reversible space-efficient TM that compares the magnitudes of the first $m$ bits of two input bit strings (i.e. of an element of the QPE+R output and $\Omega_m$) and signals whichever is greater by flipping or not flipping a bit on an output tape.
	
	With these modifications, we have constructed a reversible classical Turing machine $W'_{\text{rev}}$ that takes an $n$-bit string $\tilde{\phi}$ and a natural number $m$ written in unary as input and uses no more than $n+3$ space.
	Before the computation begins, it has a tape with a single bit $c$ written on it.  
	It halts and outputs $\neg c$ if $0^m < \tilde{\phi} \upharpoonright m < \Omega_m \upharpoonright m$ and runs forever otherwise.
	
	As reversible classical Turing machines are special cases of quantum Turing machines, both $R$ and $W'_{\text{rev}}$ are also QTMs, and so the three-stage QTM is reversible and space-efficient. 
	
	\section{Ground State Energy of Feynman-Kitaev Hamiltonian}\label{app:FKgse}
	We consider the Feynman-Kitaev Hamiltonian
	\begin{equation}
		H=H_{\prop}+H_{\inn}+H_{\out},
	\end{equation}
	where
	\begin{gather}
		H_{\prop}=\sum_{t=0}^{T-1} (\ket{t}\bra{t}\otimes\mathbbm{1}+\ket{t+1}\bra{t+1}\otimes\mathbbm{1}
		-\ket{t+1}\bra{t}\otimes U_{t+1} - \ket{t}\bra{t+1}\otimes U_{t+1}^\dagger
		),\\
		H_{\inn}=\ket{0}\bra{0}\otimes(\sum \Pi_{\inn}),\\
		H_{\out}=\ket{T}\bra{T}\otimes\Pi_{\out}.
	\end{gather}
	Here $H_{\inn}$ and $H_{\out}$ penalize wrong inputs and outputs, respectively.
	Note that while $H_{\out}$ only contains one projector without loss of generality, $H_{\inn}$ might contain several projectors so that each projector is local.
	
	Define
	\begin{equation}
		\epsilon=\max_{\ket{\phi}\in\ker(\sum\Pi_{\inn}), \ket{\eta}\in\ker \Pi_{\out}}
		|\bra{\eta}U_T\cdots U_1\ket{\phi}|^2.
	\end{equation}
	A large $\epsilon$ means there \emph{exists} a correctly initialized state that will be accepted by the computation with large probability (the ``YES" instance), while a small $\epsilon$ means \emph{any} correctly initialized state is rejected with large probability (the ``NO" instance).

	To start with, we quote the following result from \cite{Watson19} (Theorem 6.1 therein, rephrased)\footnote{
		The result in \cite{Watson19} is valid (with a slight modification) for the family of Hamiltonians called standard-form Hamiltonians. It includes the Feynman-Kitaev Hamiltonians as a special case and also includes more complicated clock constructions such as \cite{GI09,CPW}. 
		However, the analysis reduces to the Feynman-Kitaev Hamiltonian, thanks to the Clairvoyance Lemma \cite{aharonov2009power,CPW,Watson19}.
		In our paper, we will apply the results for standard-form Hamiltonians. 
	}.
	\begin{lemma}\label{gse upper bound}
		If there exists a computational path such that the final state of the computation $\ket{\psi_T}$ satisfies $					\bra{T}\bra{\psi_T}H_{\out}\ket{\psi_T}\ket{T} \leq \eta$
		then the ground state energy is bounded by 
		\begin{equation}
			0 \leq \lambda_{\min}(H) = O\left(\frac{\eta}{T^2}\right).
		\end{equation}
	\end{lemma}
	In our notation, it claims $\lambda_0(H)=O(\frac{1-\epsilon}{T^2})$.
	
	We complement this result with the following:
	\begin{theorem}\label{thm:tightgap}
		$\lambda_0(H)=\Theta(\frac{1-\epsilon}{T^2})$. 
	\end{theorem}
	\begin{proof}
		We only need to prove 	$\lambda_0(H)=\Omega(\frac{1-\epsilon}{T^2})$.
		
		Conjugating $H$ by a control unitary $W=\sum_{t=0}^{T}\ket{t}\bra{t}\otimes U_t\cdots U_1$, we can assume 
		\begin{equation}
			\begin{aligned}
				H=H_{\prop}+H_{\inn}+H_{\out}
				&=\Delta\otimes\mathbbm{1}+ \ket{0}\bra{0}\otimes(\sum\Pi_{\inn}) + \ket{T}\bra{T}\otimes\tilde\Pi_{\out}\\
			\end{aligned}
		\end{equation}
		where $\Delta$ is the Laplacian of a path graph of $T + 1$ vertices and is positive semi-definite, 
		$\tilde\Pi_{\out}=(U_T\cdots U_1)^\dagger \Pi_{\out} (U_T\cdots U_1)$.
		Since we are to lower bound $\lambda_0(H)$, we can replace $\sum\Pi_{\inn}$ by $\tilde\Pi_{\inn}$, the projector onto the orthogonal complement of $\ker(\sum\Pi_{\inn})$:
		\begin{equation}
			H\geq \tilde H=\Delta\otimes\mathbbm{1}+ \ket{0}\bra{0}\otimes\tilde\Pi_{\inn} + \ket{T}\bra{T}\otimes\tilde\Pi_{\out}.
		\end{equation}
		We omit the tildes in the following.
		
		We can convert the pair of projectors $(\Pi_{\inn}, \Pi_{\out})$ into a simpler form using a unitary that only acts on the computation registers \cite{jordan1875essai}.
		The Hilbert space of the computation registers can be decomposed into mutually orthogonal one- and two-dimensional subspaces, such that on each space, one of the following five cases happens:
		\begin{itemize}
			\item $\Pi_{\inn}=\Pi_{\out}=0$;
			\item $\Pi_{\inn}=1$, $\Pi_{\out}=0$;
			\item $\Pi_{\inn}=0$, $\Pi_{\out}=1$;
			\item $\Pi_{\inn}=\Pi_{\out}=1$;
			\item \begin{equation}\label{eq:projpairsf}
				\Pi_{\inn}=\begin{pmatrix} 1,0\\0,0 \end{pmatrix},
				~~\Pi_{\out}=\begin{pmatrix} 1-\mu,&-\xi\\-\xi,&\mu \end{pmatrix}.
			\end{equation}
			where $(1-\mu)\in(0,1)$ and $\xi=\sqrt{\mu(1-\mu)}$.
		\end{itemize}
		The above decomposition allows us to block diagonalize $H$ into submatrices of size $T+1$ or $2(T+1)$. 
		Let us analyse the spectra of $H$ within each block.
		
		For case 1, $H=\Delta$, and $\lambda_0(H)=0$.
		
		For case 2, $H=\Delta+\ket{0}\bra{0}$. 
		It is not hard to diagonalize it and show $\lambda_0(H)=2-2\cos(\frac{\pi}{2L+3})$. 
		Case 3 is similar.
		
		For case 4, $H=\Delta+\ket{0}\bra{0}+\ket{T}\bra{T}$.
		It is not hard to diagonalize it and show $\lambda_0(H)=2-2\cos(\frac{\pi}{L+2})$.
		
		For case 5, using Eq.(\ref{eq:projpairsf}) and switching the first and last indices of the second block, we get (take $T=3$ for an example):
		\begin{equation}\label{eq:FKbigmat}
			H=\left(
			\begin{array}{cccccccc}
				2 & -1 & 0 & 0 & 0 & 0 & 0 & 0 \\
				-1 & 2 & -1 & 0 & 0 & 0 & 0 & 0 \\
				0 & -1 & 2 & -1 & 0 & 0 & 0 & 0 \\
				0 & 0 & -1 & 2-\mu & -\xi & 0 & 0 & 0 \\
				0 & 0 & 0 & -\xi & 1+\mu  & -1 & 0 & 0 \\
				0 & 0 & 0 & 0 & -1 & 2 & -1 & 0 \\
				0 & 0 & 0 & 0 & 0 & -1 & 2 & -1 \\
				0 & 0 & 0 & 0 & 0 & 0 & -1 & 1 \\
			\end{array}
			\right).
		\end{equation}
		This Hamiltonian can also be exactly solved. 
		It is a quantum walk on $2T+2$ sites with an ``impurity" in the middle.
		The eigenvectors $\ket{\psi}$ satisfy the following ansatz (we index the size by $-T-1,\cdots T$):
		\begin{equation}
			\psi_i=\begin{cases}
				a_1 e^{i k x}+a_2 e^{-i k x}~~&\text{if $i<0$}\\
				b_1 e^{i k x}+b_2 e^{-i k x}~~&\text{if $i\geq 0$}\\
			\end{cases},
		\end{equation}
		where $a_i$, $b_i$ and $k$ need to be determined. 
		The corresponding eigenvalues are $E(k)=2-2\cos(k)$.
		
		Using this ansatz, $H\ket\psi=E\ket\psi$ reduces to four equations, two from the boundaries and two from the middle:
		\begin{equation}
			\left(
			\begin{array}{cccc}
				1 & e^{2 i k (T+2)} & 0 & 0 \\
				e^{i k}-\mu & e^{i k}(1-e^{i k})+e^{2 i k}(1-\mu)  & -e^{i
					k} \xi  & -e^{i k} \xi  \\
				\xi  & e^{2 i k} \xi  & -1+e^{i k} \epsilon  & e^{i k} \epsilon -e^{2 i k}
				\\
				0 & 0 & e^{2 i k (T+1)}(1-e^{i k}) & -e^{i k}(1-e^{i k}) \\
			\end{array}\right)
			\left(
			\begin{array}{c}
				a_1\\a_2\\b_1\\b_2
			\end{array}
			\right)=0.
		\end{equation}
		In order for $\ket{\psi}$ to exist, the determinate must vanish.
		Straightforward calculation shows
		\begin{equation}
			\det=z^2(1-z)^2 \left(\left(z^{2 T+3}+1\right)^2-(1-\mu) (z+1)^2  z^{2 T+2}\right)=z^2(1-z)^2p_+(z)p_-(z),
		\end{equation}
		where $z=e^{ik}$, $p_\pm(z)=z^{2T+3}+1\pm\sqrt{1-\mu} z^{T+1}(z+1)$.
		
		Although $z=1$ ($k=0$) permits non-zero $(a_1,a_2,b_1,b_2)$, the corresponding $\psi_i$ would be all zero, hence $z=1$ is not a solution. 
		We are left with analysing the factor $p_\pm(z)$. 
		Simple calculation shows
		\begin{equation}
			p_\pm(z)
			=2 e^{i k \left(T+\frac{3}{2}\right)} 
			\left(
			\cos ((T+\frac{3}{2})k)\pm\sqrt{1-\mu} \cos(\frac{1}{2}k)
			\right).
		\end{equation}
		
		It is not hard to show that it has $(T+1)$ solutions on $k\in(0,\pi)$ and $k\in(\pi,2\pi)$ respectively.
		Together with the solution $k=\pi$, all $(2T+3)$ solutions satisfy $k\in\mathbb{R}$ (in other words, all $(2T+3)$ roots of $p_\pm(z)$ are on the unit circle if $0<\mu<1$).
		Moreover, the smallest solution is achieved by $p_-(z)$, and $k_0=\Theta(\frac{\sqrt{1-\sqrt{1-\mu}} }{T})=\Theta(\frac{\sqrt{\mu} }{T})$.

		Therefore, for case 5, $\lambda_0(H)=2-2\cos(k_0)=\Theta(\frac{\mu}{T^2})$.

		Now that we have bounded $\lambda_0(H)$ in each block, let us analyse $\lambda_0(H)$ of the total Hamiltonian.
		If $\epsilon=1$, then the theorem simply says $\lambda_0(H)\geq 0$ and is true.
		If $\epsilon<1$, then subspaces of the first type do not exist.
		For type 2, 3, 4, $\lambda_0(H)=\Omega(\frac{1}{T^2})$ regardless of $\epsilon$.
		For type 5, pick $\ket{\phi}=\begin{pmatrix} 0\\1 \end{pmatrix}$ (kernel of $\Pi_{\inn}$ in Eq.(\ref{eq:projpairsf})), we find	
		\begin{equation}
			\epsilon\geq \norm{(1-\Pi_{\out})U_T\cdots U_1\ket{\phi}}^2=1-\mu.
		\end{equation}	
		Therefore, we always have $\lambda_0(H)=\Omega(\frac{1-\epsilon}{T^2})$.
	\end{proof}

	\section{Construction of the Hamiltonian}\label{appendix: hamiltonian encoding}
	
	The encoding of the three-stage QTM into a Hamiltonian can be done as in \cite{BCW}. 
	Here, we describe the key ingredients concisely and refer readers to \cite[Appendix B-F]{BCW} for extended discussions.

	\subsection{Computational Hamiltonian}
	The goal of this subsection is to construct a Hamiltonian that encodes the computation of the QTM via a history state construction. 
	It is mainly based on \cite{GI09} and \cite{CPW,CPWfull}.

	\textbf{Initialization and non-halting penalty}
	
	First, we comment on how to construct penalty terms that energetically favour the three-stage QTM to have been correctly initialized and to have halted. 
	\begin{lemma}\label{circuit}
		There exists a quantum circuit that can be wrapped around the three-stage QTM such that if the three-stage QTM halts with probability $\eta$ on $\ket{1}^{\otimes L}$, then for an initial state 
		\begin{equation}
			\ket{\psi_0} = \ket{0}_{\text{\upshape{anc}}}\left( \alpha \ket{1}^{\otimes L} + \sqrt{1-\alpha^2}\ket{\phi} \right)
		\end{equation}
		where $\ket{\phi} \perp \ket{1}^{\otimes L}$, the final state of the computation $\ket{\psi_T}$ satisfies
		\begin{equation}\label{eq:lemma6ineq}
			1-\frac{(1+\alpha\sqrt\eta)^2}{4}
			\leq
			\bra{\psi_T} (\ket{1}_{\anc}\bra{1} \otimes \mathds{1}^{\otimes L}) \ket{\psi_T}	
			\leq 
			\frac{3}{4}\left|\alpha\sqrt{1-\eta}+\sqrt{1-\alpha^2}\right|^2.
		\end{equation}
	\end{lemma}
	That is, if we assume the ability to enforce a specific ancilla be in the state $\ket{0}$ at the start of the computation, then its state at the end being $\ket{0}$ indicates both the QTM halted ($\eta=1$) and the QTM was correctly initialized ($\alpha=1$) as $\ket{1}^{\otimes L}$. 
	\newcommand{\nh}{\text{\upshape{nh}}}
	\newcommand{\hh}{\text{\upshape{h}}}
	\begin{proof}
		The circuit-QTM combination is constructed as in \cite{BCW}: a $R_{\frac{2\pi}{3}}$ ancilla rotation, following by a $R_{-\frac{\pi}{3}}$ ancilla rotation conditioned on the initial state being $\ket{1}^{\otimes L}$, followed by the QTM evolution $\mathcal{M}$, followed by a $R_{-\frac{\pi}{3}}$ ancilla rotation conditioned on QTM halts.
		A simple calculation shows:
		\begin{equation}\label{eq:lemma6decomp}
			\begin{aligned}
				\ket{\psi_T}=
				&\frac{\sqrt 3}{2}\ket{1}_{\anc}\left(
				\alpha\sqrt{1-\eta}\ket{\psi_{\nh}}+\sqrt{1-\alpha^2}(\sqrt{\eta'}\ket{\phi_\hh}-\sqrt{1-\eta'}\ket{\phi_{\nh}})
				\right)\\
				+&\frac{1}{2}\ket{0}_{\anc}\left(
				\alpha\sqrt{\eta}\ket{\psi_\hh}+\mathcal{M}
				(\alpha \ket{1}^{\otimes L} + \sqrt{1-\alpha^2}\ket{\phi})
				\right)
			\end{aligned}
		\end{equation}
		Here, we have decomposed $\ket{1}$ and $\ket{\phi}$ according to whether a branch will halt or not:
		\begin{equation}
			\begin{aligned}
				\mathcal{M}\ket{1}^{\otimes L}&=\sqrt{\eta}\ket{\psi_\hh}+\sqrt{1-\eta}\ket{\psi_{\nh}},\\ \mathcal{M}\ket{\phi}&=\sqrt{\eta'}\ket{\phi_\hh}+\sqrt{1-\eta'}\ket{\phi_{\nh}}.
			\end{aligned}
		\end{equation}
		Eq. (\ref{eq:lemma6ineq}) then follows by applying the triangle inequality to each line of Eq. (\ref{eq:lemma6decomp}).
	\end{proof}

	\textbf{Encoding the QTM transition rules}

	The following lemma claims that the evolution of the circuit\&QTM combination can be encoded into a translationally-invariant nearest-neighbour Hamiltonian
	\begin{equation}\label{eq:lemmaQTM}
		H_{\QTM}(L,\phi) := \sum_{i=1}^{L-1} h_{(i,i+1)}(\phi).
	\end{equation} 
	\begin{lemma}[\cite{CPWfull,BCW}] \label{qtm to ham}
		We can explicitly construct a Hermitian operator $h \in \mathcal{B}(\mathbb{C}^d \otimes \mathbb{C}^d)$, 
		such that
		\begin{enumerate}
			\item $h\geq 0$,
			\item $d$ depends (at most polynomially) on the alphabet size and number of internal states of $\mathcal{M}$,
			\item $h = A + e^{i\pi \phi}B + e^{-i \pi \phi}B^\dagger$, where
			\begin{itemize}
				\item $A \in \mathcal{B}(\mathbb{C}^d \otimes \mathbb{C}^d)$ is Hermitian with coefficients in $\mathbb{Z} + \mathbb{Z}/\sqrt{2}$;
				\item $B \in \mathcal{B}(\mathbb{C}^d \otimes \mathbb{C}^d)$ with coefficients in $\mathbb{Z}$.
			\end{itemize}
		\end{enumerate}			
		Furthermore, on a qudit chain of length $L$ with local dimension $d$, the Hamiltonian in Eq. (\ref{eq:lemmaQTM})	has the following properties:
		\begin{enumerate}[resume]
			\item $H_{\QTM}(L,\phi)$ is frustration-free,
			\item the ground state subspace of $H_{\QTM}(L,\phi)|_{\mathcal{S}_{\text{\upshape{\text{br}}}}}$ is spanned by the history states encoding the evolution of the QTM: $\frac{1}{\sqrt{T+1}} \sum_{t=0}^T \ket{t} \ket{\psi_t}$.
			
			Here, the bracketed subspace is defined by $\mathcal{S}_{\text{\upshape{\text{br}}}}:= \ket{<} \otimes \mathcal{H}^{\otimes (L-2)} \otimes \ket{>}$ where $\ket{<}$ and $\ket{>}$ are special states in the local Hilbert space. 
		\end{enumerate}
		The duration of the computation $T = \Omega( \text{\upshape{poly}}(L)\xi^L)$ is independent of the computation details. Here $\xi$ is an integer that can be chosen at will as long as it is sufficiently large.
	\end{lemma}

	\textbf{The computational Hamiltonian}
	
	Now we add the penalty terms to the Hamiltonian:
	\begin{equation}
		H_{\text{\upshape{\text{comp}}}}(L,\phi) := H_{\QTM}(L,\phi)  + \sum_{i=1}^{L-1} h^{\text{\upshape{pen}}}_{(i,i+1)}.
	\end{equation} 
	Note that besides the initialization and non-halting penalty discussed above, denoted by $H_{\text{out}}$,
	the penalties also contain $H_{\text{in}}$, which tries to enforce the ancilla in lemma \ref{circuit} to initialise as $\ket{0}_{\text{\upshape{anc}}}$.
	
	\begin{theorem}\label{L asympt}
		We can explicitly construct a Hermitian operator $h^{\text{\upshape{pen}}}\in \mathcal{B}(\mathbb{C}^{d}\otimes \mathbb{C}^{d})$, for $d$ from lemma \ref{qtm to ham}, such that
		\begin{enumerate}
			\item $h^{\text{\upshape{pen}}} \geq 0$,
			\item $h^{\text{\upshape{pen}}}$ is diagonal with terms in $\mathbb{Z}$.
		\end{enumerate}
		Furthermore, on a spin chain of length $L$ with local dimension $d$, the Hamiltonian $H_{\text{\upshape{\text{comp}}}}$	has the following properties:
		\begin{enumerate}
			\item if no $\tilde{\phi}_m$ satisfies $0 < \tilde{\phi}_m < \Omega_m\upharpoonright m$, then  
			\begin{equation}\label{eq:L asympt1}
				\lambda_0 (H_{\comp}(L,\phi) ) = \Omega\left(\frac{1-2^{-(n-m)}-\delta(n)}{T^2}\right);
			\end{equation}
			\item if both $\tilde{\phi}_m$ satisfy $0 < \tilde{\phi}_m < \Omega_m \upharpoonright m$, then
			\begin{equation}
				\lambda_0 (H_{\comp} (L,\phi)) = O\left(\frac{2^{-(n-m)}+\delta(n)}{T^2}\right);
			\end{equation}
			\item otherwise
			\begin{equation}
				\lambda_0 (H_{\comp} (L,\phi)) \geq 0.
			\end{equation}
		\end{enumerate}
		Here $\tilde{\phi}_m \in \mathcal{I}_m(\phi)$ for $\mathcal{I}_m(\phi)$ as in Eq. (\ref{eq:defIm}), $n=L-5$, $m\leq n$, and $\delta(n)$ is as in lemma \ref{QTM family}. 
	\end{theorem}
	\begin{proof}
		\cite{BCW} constructs $h^{\text{\upshape{pen}}}$ that satisfies points 1 and 2.
		
		By theorem \ref{thm:tightgap} we have $\lambda_0(H)=\Omega(\frac{1-\epsilon}{T^2})$ where
		\begin{equation}
			\epsilon = \max_{\ket{\psi}\in\ker(\sum\Pi_{\inn}),\ket{\eta}\in\ker(\Pi_{\out})} |\bra{\eta}U\ket{\psi}|^2.
		\end{equation}
		Here, $\Pi_{\out}$ is rank 1, 
		hence $1-\epsilon$ is exactly the minimal output penalty subjected to the constraint that  $\ket{\psi_0}\in\ker(H_{\inn})$, which is given by lemma $\ref{circuit}$.
		Therefore,
		\begin{equation}
			1-\epsilon\geq \min_{\ket{\psi_0}}\left(1-\frac{(1+\alpha\sqrt\eta)^2}{4}\right)\geq 1-\frac{(1+\sqrt{\eta})^2}{4},
		\end{equation}
		and
		\begin{equation}
			1-\epsilon\leq \min_{\ket{\psi_0}} \frac{3}{4}|\alpha\sqrt{1-\eta}+\sqrt{1-\alpha^2}|^2\leq\frac{3}{4}(1-\eta).
		\end{equation}

		For case 1, by lemma \ref{QTM family} we have $\eta \leq 2^{-(n-m)}+\delta(n)$, hence Eq. (\ref{eq:L asympt1}) follows.
		
		For case 2, by lemma \ref{QTM family} we have $\eta\geq 1-2^{-(n-m)}-\delta(n)$. 
		The claim then follows from lemma \ref{gse upper bound}.

		For case 3 (if there are two $\tilde{\phi}_m \in \mathcal{I}_m(\phi)$ for the given $m$ and exactly one is less than $\Omega_m \upharpoonright m$), the bound comes from the positive semi-definiteness of $H_{\comp}$. 
		This bound will not be used later.
	\end{proof}

	\subsection{Tiling\&Marker Hamiltonian}

	The separation in ground state energies for $H_\text{comp}$ vanishes as $L \rightarrow \infty$. 
	However, as discussed in the main text, we can couple $H_\text{comp}$ to a tiling\&marker Hamiltonian to create a separation that persists as $L \rightarrow \infty$.

	\textbf{Tiling Hamiltonian}
	
	Based on the idea of Wang tilings we can design a classical Hamiltonian, which we call a ``checkerboard Hamiltonian", such that the zero energy ground states correspond to a periodic partition of the lattice into squares of equal and unfixed size. 
	Wang tiles can also be used to design another classical Hamiltonian, which we call a ``TM Hamiltonian", such that the zero energy ground states encode the computation of any desired classical Turing machine as a tiling.
	We can then couple the checkerboard Hamiltonian and a TM Hamiltonian encoding the computation of $x\mapsto \ceil{x^{1/8}}$, so that in the ground states a marker $\bigstar$ will be placed at $\ceil{s^{1/8}}$ on the upper boundary of each square of size $s$.
	We denote the combined Hamiltonian as $H_{\tile}$.
	
	\textbf{Marker Hamiltonian}
	
	The marker Hamiltonian \cite{Bausch_2020} is a family of 1D translationally-invariant Hamiltonians on a qudit chain.
	Certain basis states of the local Hilbert space are deemed special and called \textit{marks} (in our case, they are $\ket{\blacksquare}$ and $\ket{\bigstar}$). 
	The Hamiltonian, denoted $H_{\text{mark}}$, is designed to never move these special states. 
	Furthermore, its ground state energy depends on the distance between consecutive $\blacksquare$ states and their relative positions with respect to the $\bigstar$ marks.
	In particular, we use a marker Hamiltonian $H_\mark$ that, when a $\bigstar$ is placed between two $\blacksquare$ markers---which will act as endpoint markers---the ground state energy is\footnote{In fact, using a similar method as in the solution of Eq. (\ref{eq:FKbigmat}), we can show 
		$\bra{\boxplus_s}_A  H^{(\boxplus ,f)}_s|_A \ket{\boxplus_s}_A =-\frac{1}{4^{C(s+\lceil s^{\sfrac{1}{8}} \rceil)}}(\frac{9}{4}+o(1))$.
	}	
	\begin{equation}
		-\frac{1}{4^{-C(L+r)}}\leq\lambda_0(H_\mark)\leq -\frac{3}{4^{-C(L+r)}},
	\end{equation}
	where $C\in\mathbb{N}$ is a constant that we are free to choose, $L$ is the distance between two $\blacksquare$ marks, and $r$ is the distance between the $\bigstar$ mark and the left $\blacksquare$ mark.

	\textbf{Tiling\&Marker}
	
	We consider a two-layer system composed of a marker layer and a tiling layer. 
	In the marker layer, we extend each local Hilbert space $\mathcal{H}_{\text{mark}}$ by one dimension, and place a marker Hamiltonian on each row of the lattice.
	The local Hilbert space of the two-layer system is $\mathcal{H}_\boxplus=\mathcal{H}_{\tile} \otimes (\mathcal{H}_{\mark}\oplus\mathbb{C})$.
	The two layers are coupled with penalty terms such that in the ground state of the coupled tiling\&marker Hamiltonian $H_\boxplus$:
	\begin{itemize}
		\item states in the marker Hilbert space live on the upper edge of each square in the tiling layer; 
		\item $\ket{0}\in\mathbb{C}$ lives elsewhere; 
		\item one $\blacksquare$ in the marker layer matches with one of each of the top corners of each square in the tiling layer;
		\item the $\bigstar$ mark in the tiling layer matches with the $\bigstar$ mark placed at $\ceil{s^{1/8}}$ in the tiling layer.
	\end{itemize}

	\textbf{Tiling\&Marker coupled with computation}
	
	We add one more layer---a computational layer---to the tiling\&marker system.
	In this computational layer, we extend each local Hilbert space $\mathcal{H}_{\text{comp}}$ by one dimension, and place a 1D computational Hamiltonian of lemma \ref{L asympt} on each row of the lattice.  
	The local Hilbert space of the three-layer system is $\mathcal{H}_{\text{u}}=\mathcal{H}_\boxplus\otimes(\mathcal{H}_{\text{comp}}\oplus\tilde{\mathbb{C}})$.
	The computational layer is coupled to the other two layers with penalty terms such that in the ground state of the coupled Hamiltonian $H_{\text{u}}$, the computational layer satisfies:
	\begin{itemize}
		\item ground states of the computational Hamiltonian live on the upper edge of each checkerboard square; 
		\item $\ket{\tilde 0}\in\tilde{\mathbb{C}}$ lives everywhere else. 
	\end{itemize}
	The ground state energy will be the sum of energies of each checkerboard square, each of which equals $\lambda_0(H_{\comp})+\lambda_0(H_\mark)$ living on the upper edge.

	We summarize the construction in the following lemma.
	\begin{lemma}[\cite{BCW}] \label{2D marker}
		Let $C\in \mathbb{N}$ be a sufficiently large constant. 
		We can explicitly construct Hermitian operators $h^{\site} \in \mathcal{B}\left(\mathcal{H}_{\uu}\right)$, $h^{\row},h^{\col},p^{\row} \in \mathcal{B}\left(\mathcal{H}_{\uu} \otimes \mathcal{H}_{\uu} \right)$, 
		where $\mathcal{H}_\uu = \mathcal{H}_\boxplus \otimes (\mathcal{H}_{\text{comp}}\oplus \tilde{\mathbb{C}})\cong \mathbb{C}^{d'}$ 
		and $\mathcal{H}_\boxplus=\mathcal{H}_{\tile} \otimes (\mathcal{H}_{\mark}\oplus\mathbb{C})$ 
		as discussed above,
		such that
		\begin{enumerate}
			\item $h^{\site},h^{\row},h^{\col},p^{\row} \geq 0$;
			\item $h^{\site} \in \mathcal{B}(\mathbb{C}^{d'})$ is Hermitian with coefficients in $\mathbb{Z} + \mathbb{Z}/\sqrt{2}$;
			\item $h^{\row} = B_1 + e^{i\pi\phi}B_2 + e^{-i\pi\phi}B_2^\dagger$, where
			\begin{itemize}
				\item $B_1\in \mathcal{B}(\mathbb{C}^{d'}\otimes \mathbb{C}^{d'})$ is Hermitian with coefficients in $\mathbb{Z} + \mathbb{Z}/\sqrt{2}$;
				\item $B_2 \in \mathcal{B}(\mathbb{C}^{d'}\otimes \mathbb{C}^{d'})$ has coefficients in $\mathbb{Z}$;
			\end{itemize}
			\item $h^{\col}, p^{\row} \in \mathcal{B}(\mathbb{C}^{d'}\otimes \mathbb{C}^{d'})$ are diagonal with coefficients in $\mathbb{Z}$.
		\end{enumerate}
		Furthermore on an $L\times L$ spin lattice with local dimension $d'$, the translationally-invariant nearest-neighbour Hamiltonian
		\begin{equation}
			H_{\uu}(\phi):= \sum_{i=1}^L \sum_{j=1}^L h^{\site}_{(i,j)}(\phi) + \sum_{i=1}^L \sum_{j=1}^{L-1} h^{\col}_{(i,j),(i,j+1)} + \sum_{i=1}^{L-1} \sum_{j=1}^{L} \left( h^{\row}_{(i,j),(i+1,j)}(\phi) + p^{\row}_{(i,j),(i+1,j)} \right)
		\end{equation}
		has the following properties:
		\begin{enumerate}
			\item $H_{\uu}$ block-decomposes as $(\oplus_s H_s) \oplus B$, where each $s$ is a possible square size of a valid periodic tiling of the lattice, and $B$ corresponds to all other tiling configurations.  
			\item $B \geq 0$;
			\item There exists an eigenbasis of $H_s$ consisting of states that are product states across squares in the tiling.
			\item Within a single square $A$ of side length $s$ within a block $H_s$, 
			the ground state is of the form $\ket{\boxplus_s}_A \otimes \ket{r_0}\otimes \ket{r}$, such that
			\begin{enumerate}
				\item $\ket{\boxplus_s}$ is the ground state of $H_\boxplus$, and
				\begin{equation}\label{eq:markerE}
					-\frac{3}{4^{C(s+\lceil s^{\sfrac{1}{8}} \rceil)}} \leq \bra{\boxplus_s}_A  H_\boxplus|_A \ket{\boxplus_s}_A \leq -\frac{1}{4^{C(s+\lceil s^{\sfrac{1}{8}} \rceil)}};
				\end{equation} 
				\item $\ket{r_0}$ is the ground state of $H_{\comp}$;
				\item $\ket{r}=\ket{\tilde{0}}^{\otimes (s \times (s-1))}$.
			\end{enumerate}
		\end{enumerate}
	\end{lemma}

	\subsection{Constant Spectral Gap}
	\textbf{Fine-tuning the parameters}
	
	The following lemma shows that if we set the QPE parameter $m$ to so that $\lim_{s\to\infty}m=\infty$ monotonically and
	\begin{equation}\label{eq:setm}
		m<n-(n^{1/4}-2\log_2 n),
	\end{equation}
	where $n=s-5$ and $s$ is the size of the lattice $H_{\comp}$ is running on\footnote{In the Hamiltonian from lemma \ref{2D marker}, $s$ no longer need be directly related to the size $L$ of the overall lattice}, then the constant $C$ in lemma \ref{2D marker} can be chosen such that the energy contribution in lemma \ref{2D marker}(1.d) asymptotically lies between the ground state energy bounds for $H_{\comp}$.

	\begin{lemma}\label{f(L) lemma}
		If Eq. (\ref{eq:setm}) holds, then there exists a constant $C\in\mathbb{N}$ such that
		\begin{equation}
			\omega\left(\frac{2^{-(n-m)}+\delta(n)}{T^2}\right) 
			\leq
			\frac{1}{4^{C(s+\lceil s^{1/8} \rceil)}}
			\leq o\left(\frac{1-2^{-(n-m)}-\delta(n)}{T^2}\right)
		\end{equation}
		for $\delta(n)$ from lemma \ref{QTM family}.
	\end{lemma}
	\begin{proof}
		Since $m<n$ and $\delta(n)\to 0$ as $n\to\infty$, we have:
		\begin{equation}
			\frac{1-2^{-(n-m)}-\delta(n)}{T^2} 
			= \Theta\left(\frac{1}{T^2}\right).
		\end{equation}
		Moreover, Eq. (\ref{eq:setm}) implies that $2^{-(n-m)}=O(n^22^{-n^{1/4}})$.
		Since $c_1<4$ and $c_2\geq 1$, we also have $\delta(n)=O(n^22^{-n^{1/4}})$.
		Therefore:
		\begin{equation}
			\frac{2^{-(n-m)}+\delta(n)}{T^2}
			=O\left(\frac{n^22^{-n^{1/4}}}{T^2}\right).
		\end{equation}
		By construction, there exists a constants $\xi \in \mathbb{N}$ such that
		\begin{equation}
			\log T=s\log\xi+\Theta(\log s).
		\end{equation}
		We pick $C=\log_2\xi$, which can be chosen as an integer, then
		\begin{equation}
			2\log T+\omega(1)
			\leq
			\log4\cdot C(s+\lceil s^{1/8} \rceil) 
			\leq
			2\log T-\log(n^22^{-n^{1/4}})-\omega(1).
		\end{equation}
		This implies that the same $C$ satisfies
		\begin{equation}
			\omega\left(\frac{n^22^{-n^{1/4}}}{T^2}\right)
			\leq
			\frac{1}{4^{C(s+\lceil s^{1/8} \rceil)}}
			\leq				o\left(\frac{1}{T^2}\right), 
		\end{equation}
		and the claim follows.
	\end{proof}

	\textbf{Positive/negative ground state energy density separation}
	
	We take the Hamiltonian from lemma \ref{2D marker}, set $C$ to be the constant from lemma \ref{f(L) lemma} and take $m$ chosen according to Eq. (\ref{eq:setm}). 
	The following lemma claims a positive/negative ground state energy density separation.

	\begin{lemma}\label{block gse}
		There exists a constant computable integer $s'$, such that for a square of size $s$, the Hamiltonian $H_s$ in lemma \ref{block gse} satisfies:
		\begin{enumerate}
			\item if $s<s'$ then $H_s \geq 0$;
			\item if $s\geq s'$, then
			\begin{enumerate}
				\item if no $\tilde{\phi}_m$ satisfies $0< \tilde{\phi}_m < \Omega_m \upharpoonright m$ then $H_s \geq 0$;
				\item if both $\tilde{\phi}_m$ satisfy $0< \tilde{\phi}_m < \Omega_m \upharpoonright m$ then $\lambda_0 (H_s) < 0$;
				\item otherwise $\lambda_0 (H_s) \geq \frac{-3}{4^{C(s+\ceil{s^{1/8}})}}$.
			\end{enumerate}
		\end{enumerate}
		Here $\tilde{\phi}_m \in \mathcal{I}_m(\phi)$ for $\mathcal{I}_m(\phi)$ as in Eq. (\ref{eq:defIm}).
	\end{lemma}
	\begin{proof}
		We have
		\begin{equation}
			\lambda_0(H_s)=\lambda_0(H_{\comp})+\bra{\boxplus_s}_A  H_\boxplus|_A \ket{\boxplus_s}_A.
		\end{equation}
		
		Lemma \ref{f(L) lemma} implies that the second term lies between two cases from theorem \ref{L asympt} asymptotically.
		We define $s'$ to be the constant such that the separation holds for all $s\geq s'$.
		Since all implicit constants in relevant proofs are computable, so is $s'$.
		
		We can design the marker Hamiltonian such that if $s<s'$ then it is positive semidefinite. Consequently $H_s\geq 0$. 
		In the following, we assume $s>s'$. 
		
		(a) In this case, theorem \ref{L asympt}(1) implies $\lambda_0(H_{\comp})=\Omega(\frac{1-2^{-(n-m)}-\delta(n)}{T^2})$. Eq. (\ref{eq:markerE}) and lemma \ref{f(L) lemma} imply the 2nd term $=-o(\frac{1-2^{-(n-m)}-\delta(n)}{T^2})$.
		Therefore, $\lambda_0(H_s)> 0$.
		
		(b) In this case, theorem \ref{L asympt}(2) implies
		$\lambda_0 (H_{\comp}) = O(\frac{2^{-(n-m)}+\delta(n)}{T^2})$ while Eq. (\ref{eq:markerE}) and lemma \ref{f(L) lemma} imply the 2nd term $=-\omega(\frac{2^{-(n-m)}+\delta(n)}{T^2})$. 
		Therefore, $\lambda_0(H_s)<0$.
		
		(c) In this case, theorem \ref{L asympt}(3) implies $\lambda_0(H_{\comp})\geq 0$. Combined with  Eq. (\ref{eq:markerE}), we are done.	
	\end{proof}
	
	\textbf{Ground state energy of $H_{\uu}$}
	
	When taking the thermodynamic limit, this translates into a diverging separation in ground state energies because by lemma \ref{2D marker} the ground state energy is a sum of the ground state energy of each block.
	
	\begin{theorem}
		For the Hamiltonian in lemma \ref{block gse},
		\begin{equation}
			\lim_{L \rightarrow \infty} \lambda_0(H_{\uu}) \begin{cases}
				\rightarrow -\infty & \text{ for } \phi < \Omega\\
				\geq 0 & \text{ for } \phi \geq \Omega
			\end{cases}.
		\end{equation}
	\end{theorem}
	\begin{proof}

		(1) If $\phi < \Omega$, then due to theorem \ref{thm: W'transition}, there exists an $m$ such that all $\tilde{\phi}_m\in \mathcal{I}_m(\phi)$ satisfy $0< \tilde{\phi}_m < \Omega_m$. 
		Then by lemma \ref{block gse}, $\lambda_0(H_s)<0$.
		Since a partition of the whole system into size-$s$ squares is a legal eigenstate, we have: 
		\begin{equation}
			\lim_{L\rightarrow \infty} \lambda_0(H_{\uu}) \leq \lim_{L\rightarrow \infty} \sum_{\text{complete $s$-squares}} \lambda_0(H_s) \to -\infty.
		\end{equation}
		Since the right hand side is already $-\infty$, the equation holds even if there are other $m$ giving rise to other eigenstates with smaller eigenvalues.

		(2) Now consider $\phi \geq \Omega$.
		Due to theorem \ref{thm: W'transition}, there can be no $s$ such that a $\tilde{\phi}_m\in \mathcal{I}_m(\phi)$ satisfies $0< \tilde{\phi}_m < \Omega_m$. 
		By lemma \ref{block gse} (1) and (2a), we know $\lambda_0(H_s)>0$ for any $s$.
		So $\lim_{L\rightarrow \infty} \lambda_0(H_{\uu}) \geq 0$.
	\end{proof}

	\textbf{Spectral gap}

			Using techniques from \cite{Bausch_2020}, the diverging difference in ground state energy can be converted into a difference in spectral gaps. 
			
			We now take a Hamiltonian $H_{\trivial}$ with local Hilbert space $\mathcal{H}_{\trivial}\cong\mathbb{C}^2$ that is diagonal in the computational basis, has $\ket{0}^{\otimes\Lambda}$ as its unique gapped ground state whose energy equals $-L$, and has a spectral gap of 1. 
			For example, $H_{\trivial}= -\sum_{i\in \Lambda}\ket{0}\bra{0}_i+\frac{1}{2}\sum_{\braket{i,j}} \mathds{1}^{(i)}\otimes \mathds{1}^{(j)}$ would suffice, where $\sum_{\braket{i,j}}$ means summing over neighbouring lattice sites $i$ and $j$. 
			
			We also take a Hamiltonian $H_{\dense}$ with local Hilbert space $\mathcal{H}_{\dense}$ with continuous spectrum in $[0,\infty)$ in the thermodynamic limit. For example, the Hamiltonian describing 1D XY chains placed on each row with no interaction between columns would suffice.
			
			Then, consider a local Hilbert space $\mathcal{H}_{\tot} := (\mathcal{H}_\uu\otimes \mathcal{H}_{\dense})\oplus \mathcal{H}_{\trivial}$. 
			We can define the following translationally-invariant, nearest-neighbour Hamiltonian on $\otimes_\Lambda \mathcal{H}_{\tot}$:
				\begin{equation}
					\begin{aligned}
						H_\tot(\phi) :=& \beta\left( (H_\uu(\phi) \otimes \mathds{1} + \mathds{1} \otimes H_{\dense})  \oplus 0 \right)
						+ (\left( 0\otimes 0 \right) \oplus H_{\trivial})\\
						+&\sum_{\braket{i,j}}((\mathds{1}_\uu^{(i)} \otimes \mathds{1}_{\dense}^{(i)})\oplus 0^{(i)})\otimes ((0^{(j)}\otimes 0^{(j)})\oplus \mathds{1}_{\trivial}^{(j)})+(i\leftrightarrow j).
					\end{aligned}
				\end{equation}
				Here $\beta\in\mathbb{Q}^+$.
				The second line ensures all states receive an energy penalty except those fully supported on $\otimes_\Lambda(\mathcal{H}_\uu\otimes \mathcal{H}_{\dense})\oplus (\otimes_\Lambda\mathcal{H}_{\trivial})$. 
				Therefore, 
				\begin{equation}
					\spec(H_\tot(\phi)) = \beta(\spec(H_{\uu}(\phi))+\spec(H_{\dense}))\cup \spec(H_{\trivial}) \cup  G,
				\end{equation}
				where $G \subset [1,\infty)$. 
				
				\begin{theorem}
					$H_\tot(\phi)$ constructed above satisfies theorem \ref{theorem 1}.
				\end{theorem}
				\begin{proof}
					If $\phi<\Omega$, then $\lambda_0(H_\uu(\phi))\rightarrow -\infty$ quadratically, which implies $\lambda_0(H_\tot(\phi))\rightarrow -\infty$. 
					The low energy physics is independent of $H_\trivial$.
					The ground state is gapless as guaranteed by $H_{\dense}$.
					
					If $\phi \geq \Omega$, then $\lambda_0(H_\uu(\phi))\geq 0$, which implies $\lambda_0(H_\tot(\phi)) = \lambda_0(H_\trivial(\phi))=-L$ and a spectral gap of 1, independent of $\beta$ (hence $\beta$ can be chosen arbitrarily small). The ground state is the gapped ground state of $H_{\trivial}$. 
				\end{proof}

\end{document}